\documentclass[11pt]{article}

\usepackage[left=1in, right=1in, top=1in, bottom=1in]{geometry}
\usepackage[utf8]{inputenc}
\usepackage{amsmath,amsfonts}
\usepackage{gensymb}
\usepackage{graphicx}
\usepackage{tikz}
\usetikzlibrary{positioning}
\usepackage{mathtools}
\usepackage{amsthm}
\usepackage[numbers]{natbib}
\usetikzlibrary{arrows,automata,calc}
\usepackage{nccmath}
\usepackage{xcolor}
\definecolor{MyBlue}{rgb}{0.12, 0.12, 0.76}
\usepackage[colorlinks,allcolors=MyBlue]{hyperref}
\usepackage{multicol}
\usepackage{xspace}

\usepackage{thmtools}
\usepackage{thm-restate}
\declaretheorem[name=Theorem,numberwithin=section]{thm}
\newtheorem{theorem}[thm]{Theorem}
\newtheorem{lemma}[theorem]{Lemma}

\newtheorem{proposition}[theorem]{Proposition}
\newtheorem{definition}{Definition}[section]
\newtheorem{lemcorollary}{Corollary}[lemma]
\newtheorem{thmComment}{Comment}[section]

\DeclareMathOperator*{\argmax}{arg\,max}

\newcommand\ben[1]{\textbf{\textcolor{red}{(Ben): {#1}}}}

\newcommand\ashish[1]{\textbf{\textcolor{green}{(Ashish): {#1}}}}

\newcommand\x{\mathbf{x}}
\newcommand\y{\mathbf{y}}
\newcommand\z{\mathbf{z}}

\newcommand\B{\mathcal{B}}
\newcommand\Rfrom{R^{\leftarrow}}

\newcommand\bbR{\mathbb{R}}
\newcommand\T{\mathcal{T}}
\newcommand\calh{\mathcal{H}}

\newcommand\tat{t\^{a}tonnement\xspace}

\allowdisplaybreaks[1] 
\begin{document}
\title{Markets for Public Decision-making}

\author{Nikhil Garg \and Ashish Goel \and Benjamin Plaut}
\date{\{nkgarg, ashishg, bplaut\}@stanford.edu\\ Stanford University}

\maketitle

\begin{abstract}

A public decision-making problem consists of a set of issues, each with multiple possible alternatives, and a set of competing agents, each with a preferred alternative for each issue. We study adaptations of market economies to this setting, focusing on binary issues. Issues have prices, and each agent is endowed with artificial currency that she can use to purchase probability for her preferred alternatives (we allow randomized outcomes). We first show that when each issue has a single price that is common to all agents, market equilibria can be arbitrarily bad. 

This negative result motivates a different approach. We present a novel technique called \emph{pairwise issue expansion}, which transforms any public decision-making instance into an equivalent Fisher market, the simplest type of private goods market. This is done by expanding each issue into many goods: one for each pair of agents who disagree on that issue. We show that the equilibrium prices in the constructed Fisher market yield a \emph{pairwise pricing equilibrium} in the original public decision-making problem which maximizes Nash welfare. More broadly, pairwise issue expansion uncovers a powerful connection between the public decision-making and private goods settings; this immediately yields several interesting results about public decisions markets, and furthers the hope that we will be able to find a simple iterative voting protocol that leads to near-optimum decisions.
\end{abstract}

%

\section{Introduction}\label{sec:intro}

Fair and transparent public decision-making is a key element of a democratic society, but many public decisions are made by government officials behind closed doors. In this paper, we investigate mechanisms for large-scale public decision-making where citizens directly vote on a set of issues at the same time, focusing on the case where each issue has exactly two alternatives. In particular, we examine connections to \emph{private goods} allocation (i.e., standard resource allocation). One can think of each issue that is under consideration as a ``good", and public decision-making as ``allocating" the good to one of the alternatives. We allow randomized outcomes, where the outcome can put nonzero probability on multiple alternatives: this is analogous to \emph{divisible} private goods, where a good can be split among multiple agents\footnote{An alternative interpretation is that the issues themselves are divisible: for example, in the case of a city choosing how much money to a particular project, any amount of money is a valid outcome.}. The fundamental difference is that in private goods allocation, each agent's utility depends only on the bundle of goods she receives; in public decision-making, the group makes a single decision that affects all agents. 

Market economies are one of the longest-studied areas in the distributions of private goods. The simplest market model is that of a \emph{Fisher market} (see~\cite{brainard_how_2005} for a modern exposition), which consists of a set of available goods and a set of agents with money they wish to spend. It is typically assumed that the agents have no value for the money itself, and wish to spend their entire budget to acquire goods that they desire. Each good has a \emph{price} which determines how much money an agent must spend to purchase one unit of the good. The goods are divisible, meaning that an agent can purchase any fraction of a good. A \emph{market equilibrium} assigns a price to each good such that every agent purchases their favorite bundle that is affordable under the prices, and that the demand meets the supply. Under some mild conditions (elaborated on in Section~\ref{sec:model}), a market equilibrium always exists~\cite{arrow_existence_1954}.

\subsection{Our contribution}

We consider adaptations of markets to the public decision-making setting. Many democracy theorists believe that it is unethical (eg. see~\cite{satz_why_2012}) and many democratic countries stipulate that it is illegal to allow citizens to purchase political influence with actual money. Instead, we think of each agent being endowed with the same amount of ``artificial currency" that is useful only for voting on these issues; thus our approach to public decision markets is consistent with the spirit of ``one person one vote". Prices are assigned to issues, and agents can use their artificial currency to ``purchase" probability for their preferred alternatives on the issues they most value\footnote{This can also be thought of as a private goods market with externalities: each agent's utility depends not only on her own bundle, but also other agents' bundles.}.

Markets have the desirable property that each agent can choose how to allocate her money across goods, based on their relative values to her. In the context of large-scale public decision-making, this allows agents to express their relative weights for the different issues in a fine-grained way. This is in contrast to approaches like asking agents to rank the issues by importance, which are more limited in expressiveness. Markets have the additional property that the equilibria are ``supported" by prices: prices provide a sort of certificate of fairness, in that each agent can verify that she is spending her budget in the best way possible.

The simplest pricing model assigns a single price to each issue, and all agents are subject to the same set of prices. We refer to this as ``per-issue pricing", or just ``issue pricing". In the private goods setting, per-good pricing is sufficient to yield a market equilibrium with optimal \emph{Nash welfare}: the product of agent utilities\footnote{The concept of Nash welfare is due to \cite{nash_bargaining_1950} and \cite{kaneko_nash_1979}.}. Unfortunately, we show in Section~\ref{sec:per-issue-lower} that issue pricing in the public decisions setting can result in very poor equilibria: the Nash welfare of the equilibrium may be a factor of $O(n)$ worse than optimal, where $n$ is the number of agents. The same instance shows that the utilitarian welfare (the sum of agent utilities) and egalitarian welfare (the minimum agent utility) may both be a factor of $O(n)$ worse than optimal as well.

\subsubsection{Pairwise issue expansion.}

This negative result motivates a more complex market model. Our main contribution is a reduction which transforms any public decision-making instance into a private goods Fisher market instance that is ``equivalent" in a strong sense. For each issue, we construct a good for each pair of agents who disagree on that issue. The outcome on that issue can be thought of as the result of pairwise negotiations between each pair of agents who disagree. We refer to this reduction as \emph{pairwise issue expansion}. The equilibrium prices of the constructed Fisher market yield a ``pairwise pricing equilibrium" in the original public decisions instance. We show that the resulting pairwise pricing equilibrium maximizes Nash welfare in the public decisions instance.

Furthermore, pairwise issue expansion allows us to directly import results for Fisher markets to the public decisions setting. If the utilities in the public decisions instance are in class $\calh$ (say, linear utilities), the utilities in the constructed Fisher market will be nested $\calh$-Leontief (for example, nested linear-Leontief)\footnote{These utility classes will be defined and discussed later.}. This means that any result which works for Fisher markets with nested $\calh$-Leontief utilities can be imported to public decisions instances with utilities in class $\calh$. The main Fisher market results we consider are: (1) a strongly polynomial-time algorithm for finding a Fisher market equilibrium with two agents and any utility functions~\cite{chakrabarty_2006_rationality}, (2) a strongly polynomial-time algorithm for finding a Fisher market equilibrium for Leontief utilities with weights in $\{0,1\}$ \cite{garg_2005_primal}, (3) a polynomial-time algorithm for a Fisher market with Leontief utilities which yields a $O(\log n)$ approximation simultaneously for all canonical welfare functions (i.e. Nash welfare, utilitarian welfare, egalitarian welfare, etc)~\cite{goel_2014_price}, and (4) a discrete-time \tat process for finding the Fisher market equilibrium for nested CES-Leontief utilities that converges in polynomial-time~\cite{avigdor-elgrabli_convergence_2014}. Thus pairwise issue expansion yields the following results for the public decision-making setting:

\begin{enumerate}

\item A strongly polynomial-time algorithm for finding a public decisions market equilibrium with two agents and any utility functions.

\item A strongly polynomial-time algorithm for finding a public decisions market equilibrium for Leontief utilities with weights in $\{0,1\}$.

\item A polynomial-time algorithm for a public decisions instance with Leontief utilities which yields a $O(\log n)$ approximation simultaneously for all canonical welfare functions (i.e. Nash welfare, utilitarian welfare, egalitarian welfare, etc).

\item A discrete-time \tat process for finding a public decisions market equilibrium for CES utilities that converges in polynomial-time.

\end{enumerate}

These Fisher market results yield the analogous results for public decisions instances for two agents with any utilities, Leontief utilities with weights in $\{0,1\}$\footnote{A nested Leontief-Leontief function is still a Leontief function. Incidentally, this also implies that for a public decision making problem where agent utilities are Leontief, we get a Fisher market which has exactly the same form, i.e. with Leontief utilities.}, any Leontief utilities, and CES utilities, respectively. We also discuss public decisions \tat in more depth, and show how our reduction can be used to implement a \tat process where agents only interact with the public decisions instance, and never see the constructed Fisher market.

%

More broadly, our work uncovers a powerful connection between private goods allocation and public decision-making. We hope that pairwise issue expansion will have applications in future work as well. One particularly promising direction is to design an iterative local voting scheme akin to prediction markets~(see \cite{chen_designing_2010}), where agents (or pairs of agents) arrive sequentially and move the current decision vector in their preferred direction subject to offered issue prices. Our proof of the existence of a simple \tat for public decision markets offers hope that such a scheme may be possible.

\subsection{Related work}\label{sec:related}

Equilibrium theory has a long history in economics~\cite{arrow_existence_1954, brainard_how_2005,  varian_equity_1974, walras_elements_1874}. More recently, this topic has garnered significant interest in the computer science community as well (see~\cite{vazirani_2007_combinatorial} for an algorithmic introduction).

\paragraph{Foley's work on Lindahl Equilibria} The market concept most directly relevant to our public decision markets is that of Lindahl equilibria, developed by Foley~\cite{foley_lindahls_1970}, who showed that personalized prices (i.e., each agent may be assigned a different price for each good with no restrictions) can support any Pareto optimal solution in the context of public goods\footnote{In public goods, all agents have nonnegative utility for every good, and the question is how to allocate their money between the goods. In contrast, in public decision-making, agents have opposing preferred alternatives and are in direct competition on each issue. With a careful modification, Foley's work does carry over to the public decision-making setting.}. Our work can be thought of as improving upon Foley's work to get much stronger properties for the special case of public decision-making. We obtain these stronger properties using a more sophisticated reduction, one which is in fact weaker in the sense that there is a correspondence between the public decisions market and the private market \emph{only at equilibrium}. Our reduction explicitly relies on the fact that agents are in opposition on each issue in the public decisions setting, which is not the case in the public goods setting.

The stronger properties we obtain are as follows. First, Foley's work~\cite{foley_lindahls_1970} allows arbitrary personalized prices, whereas we only require pairwise prices: for each issue, there is a price for each pair of agents who disagree on that issue. Our Fisher market can be thought of as negotiating independently with each person that disagrees with you through a normal market; we are not aware of any such simple interpretation that follows from Foley's very general work. Second, in our private goods reduction, a feasible public goods decision (where each agent shares the same societal decision) emerges naturally: we leverage properties of nested-Leontief utilities and the Nash welfare objective function to implicitly represent the feasibility constraints, which allows us to obtain the correspondence \emph{only at equilibrium}. In contrast, Foley adds cone constraints to a private goods market to explicitly enforce the feasibility constraints of the public decision-making problem; these constraints have no natural real-world analogue. Third, we reduce the public goods setting to a Fisher market, arguably the simplest possible and most-studied private goods market. Because of this, our reduction allows us to lift many Fisher market equilibrium results to the public decision-making setting. In particular, our reduction allows us to obtain a \tat for public decision-making, even though intermediate steps in the \tat are in a regime where the public and private markets are not in direct correspondence. It is unclear whether this is possible with Foley's construction. We discuss this in technical detail and elaborate on how our work relates to Lindahl equilibria in Appendix~\ref{sec:foley}.

\paragraph{T\^{a}tonnement} As mentioned above, one of our results is a \tat for public decision-making. A \tat is an iterative process which presents agents with a set of prices, asks what they would buy given those prices, and updates the price of each good based on the aggregate demand of each good. A \tat-like process for computing the maximum Nash welfare outcome in participatory budgeting (see e.g.~\cite{goel_knapsack_2016} for more on PB) was recently given by Fain et al.~\cite{fain_core_2016}. They showed that the maximum Nash welfare outcome can be computed by using a stochastic gradient descent style algorithm. Their algorithm iteratively elicits agents' demands using a process very similar to quadratic voting~\cite{lalley_nash_2014} and updates the current solution accordingly. While this is similar to a \tat, there is one crucial difference. A true \tat (such as the one we present) allows the agents to directly change the current point: the price of each good is updated by a fixed rule based on the aggregate demand of that good. In contrast, the algorithm of \cite{fain_core_2016} moves to a point that is {\em different} from the one elicited by the quadratic voting. Also, their result also holds only for linear utilities\footnote{This discussion is thanks to Kamesh Munagala via private correspondence.}.

A \tat, with a similar elicitation scheme, has been shown to work in practice in the participatory budgeting setting~\cite{garg_collaborative_2017,garg_iterative_2018}. In those works, a new budget is directly elicited from voters, and the mechanism works for $\ell_p$ normed cost functions. However, 
their mechanism finds a total welfare maximizing point as opposed to a Nash welfare maximizing outcome. One direction for future work is to adapt the \tat from this work into such an implementable mechanism with a large number of voters.

\paragraph{Inefficiencies of pricing schemes} Another relevant paper from the economics literature is \cite{danziger_graphic_1976}, which shows that per-good pricing can lead to inefficiency for public goods. Their examples do not provide bounds on how much worse per-good pricing can be: in contrast, we show that for linear utilities, issue pricing can be a factor of $O(n)$ worse than optimal. Also, we note that it is easy to adapt the examples in Section 3 to show that two other popular market-based approaches, namely Quadratic Voting \cite{lalley_nash_2014} and Trading Post Prices \cite{shapley_trade_1977}, also do not result in good equilibria with issue pricing in our public decision market setting.

\paragraph{Strategic agents} A key property in mechanism design is \emph{strategy-proofness} (or lack thereof). A mechanism is strategy-proof if even a selfish agent would always honestly report her preferences. Most relevant to us is~\cite{schummer_strategy_1996}, which shows that even for two agents with linear utilities over divisible goods, any mechanism which is both strategy-proof and Pareto optimal\footnote{An outcome is Pareto optimal if there is no way to improve the utility of any agent without hurting another agent.} is \emph{dictatorial}, meaning that one agent receives all of the resources\footnote{A similar result holds for indivisible goods~\cite{klaus_strategy_2002}.}. Our binary-issue public decisions setting generalizes the two agent private goods setting, and hence we immediately inherit this impossibility result: any mechanism which is both strategy-proof and Pareto optimal is dictatorial. A dictatorial solution is clearly not desirable, and we would like our outcomes to be Pareto optimal, so we assume throughout this paper that agents honestly report their preferences and do not address the issue of strategic behavior.  Other incentive compatibility results for implementation of general classes of social choice functions are discussed in~\cite{dasgupta_implementation_1979}. 

We note that several works extending Foley relax the assumption that agents report their preferences truthfully, by building voting games in which the equilibrium is one in which truthful reporting is incentive compatible for each agent~\cite{groves_optimal_1977,walker_simple_1981,kaneko_ratio_1977}. Most notably, Groves and Ledyard~\cite{groves_optimal_1977} construct an allocation-taxation scheme -- using message passing -- for a market with both private goods and public commodities, such that the equilibrium behavior results in a Pareto optimal solution. As in our work, however, their mechanism is still susceptible to a manipulation in which a consumer considers how future prices and the behavior of others are affected by her current decisions.


\paragraph{Other voting schemes, and one person one vote} Other works also propose alternate voting schemes for multiple issues. Storable Votes~\cite{casella_2005_storable} allows members of a committee to store votes for future meetings so as to spend their votes on issues that matter most to them; the work proves welfare gains in the case of two voters but does not give a principled way to balance the relative importance or cost of different issues, as we do here. In \cite{conitzer_2017_fair}, the authors study adaptations of private goods fairness notions (such as proportionality) to the public decision-making context when randomized outcomes are not allowed. In contrast, we allow fractional solutions (i.e. randomized outcomes) and exactly maximize Nash welfare. 

Such works, especially this one, may seem to violate the principle of One Person One Vote~\cite{gabis_one_1978,hayden_false_2003,briffault_who_1993,ansolabehere_end_2008,gersbach_why_2004,karabarbounis_one_2011}. In particular, as we propose individual prices, a given issue may ``cost'' more for one voter than for another. However, as discussed below, these prices are generated in a principled manner -- for each issue, there is a single price for each \textit{pair} of voters who disagree on the issue. Furthermore, we note that, at the onset, each voter is allocated the same ``budget'' through which to vote on issues.

\paragraph{Other works in market equilibria for public goods} Finally, we note that many more strands of literature, too many to detail here, discuss and extend the work of Foley~\cite{foley_lindahls_1970} and more generally the idea of equilibria for the funding of public goods. This work includes both stronger results in a more specific model, as this work, and computational hardness analysis of equilibria theory in general~\cite{richter_non-computability_1999,velupillai_algorithmic_2006}. To our knowledge, our public decision-making setting has not been studied as a special case of such public goods markets.

In~\cite{bowen_interpretation_1943} and~\cite{bergstrom_when_1981}, the authors ask what happens when the decision to fund a single public good is simply made through a majority vote; in particular, they study under what conditions of voter preferences for the public good and distribution of tax shares of each voter is the funding of the good Pareto optimal. They find that majority vote can fall short of optimal if income is asymmetrically distributed. In Section~\ref{sec:per-issue-lower}, we show that the case with multiple public decisions is far worse: a generalization of majority vote -- where each issue has a price -- leads to highly suboptimal outcomes, even when everyone is endowed with the same income. 

Another strand aims to study the implications of relationships between individuals. For example, in~\cite{ray_coalitional_2001}, agents are allowed to form coalitions through binding contracts, resulting in inefficiencies. In~\cite{elliott_network_2013},  there are people who can ``produce'' a given public good and those who ``benefit'' from that good. These relationships can be represented by a network in a certain way, and the Lindahl outcomes correspond to a solution characterized by the eigenvector centralities of each node. In this work, voters who agree on a given issue end up on the same side of a bipartite graph, resulting in them purchasing the same probability for that issue.

The assumptions and philosophical underpinnings of equilibria theory are also well-studied, as are applications to other fields.
Sen~\cite{sen_rational_1977} challenges the notion that people have consistent preferences that can be elicited. In particular, he posits that people have ``commitments'' to a particular social group of other people, whose welfare they care about. We note that the assumption of a utility function is nevertheless common, though it is important to be aware of the limitations of such behavioral abstractions. In~\cite{riley_justice_1989}, competitive equilibria is connected to the idea that in capitalism people are given fruits commensurate to their labor, as part of a discussion of the relationship of notions of justice and capitalism. General equilibrium theories are even connected to Structuralism within the philosophy of science~\cite{hands_structuralist_1985}. One prominent application of the economics of public goods has been to study environmental (non\nobreakdash-)cooperation~\cite{chander_core-theoretic_2006,maler_chapter_1985}. Our work extends such applications by connecting market equilibria ideas to voting on different issues in a fair and efficient way, as discussed above.
\\\\The rest of the paper is organized as follows. Section~\ref{sec:model} formally defines the models of private goods allocation, public decision-making, and Fisher markets. Section~\ref{sec:per-issue-lower} shows that issue pricing can result in (very) poor equilibria for public decisions markets. Section~\ref{sec:reduction} presents the concept of pairwise issue expansion, and shows how this can be used to obtain optimal equilibria, as well as other properties. Section~\ref{sec:applications} gives examples of Fisher market results that we can import to the public decisions setting using our reduction. Section~\ref{sec:tat} focuses on a particular such result: \tat. Finally, Appendix~\ref{sec:foley} discusses the connection to~\cite{foley_lindahls_1970} and Lindahl equilibria in more depth. Proofs are also deferred to the appendix. 

\section{Model}\label{sec:model}

We first introduce general notation that applies to both private goods allocation and public decision-making. As much as possible, we intentionally use the same notation for the private and public settings, as one of our primary contributions is to highlight the connections between these. We then discuss our assumptions on utility functions. Finally we discuss aspects specific to private goods and specific to public decision-making.

Let $[k]$ denote the set $\{1,2 , \dots, k\}$. A \emph{private goods instance} consists of a set of agents $N = [n]$ and a set of goods $M = [m]$; a \emph{public decisions instance} consists of a set of agents $N = [n]$ and a set of issues $M = [m]$. We will typically use $i$ and $k$ to denote agents, and $j$ and $\ell$ to denote goods/issues. We assume that issues are binary, meaning that each issue $j$ has two alternatives: 0 and 1. Each agent $i \in N$ has a preferred alternative for each issue $j$, denoted by $a_{ij}$, which they truthfully report.

We assume that goods/issues are divisible, meaning that a single good/issue can split among multiple agents. In a public decision instance, divisibility can be interpreted as randomization over alternatives. An outcome of a private goods instance is an \emph{allocation} $\x \in [0,1]^{m\times n}$, where $x_i \in [0,1]^m$ is the \emph{bundle} given to agent $i$, and $x_{ij} \in [0,1]$ is the fraction of good $j$ given to agent $i$. An allocation cannot allocate more than the available supply\footnote{Although the entire supply is typically allocated, it is standard in the private goods literature to allow for outcomes where this does not occur, i.e. $\sum_{i \in N} x_{ij} < 1$. This will be discussed in Section~\ref{sec:model-market}.}: $\x$ is valid only if $\sum_{i \in N} x_{ij} \leq 1$ for all $j\in M$. The outcome of a public decisions instance is denoted by $\z = (z^1, \dots, z^m) \in [0,1]^{m\times 2}$, where $z^j = (z^{j,0}, z^{j,1}) \in [0,1]^2$, and $z^{j,a} \in [0,1]$ is the probability put on alternative $a$ for issue $j$. An outcome $\z$ is valid only if $\sum_{a \in \{0,1\}} z^{j,a} \leq 1$ for all $j\in M$.

 
\subsection{Utility functions}
\label{sec:utilityfuncs}
In a private goods instance, we use $u_i(\x) \in \mathbb{R}$ to denote $i$'s utility for allocation $\x$; in a public decisions instance,  we use $u_i(\z) \in \bbR$ to denote $i$'s utility for outcome $\z$. In a private goods instance, it is assumed that an agent's utility depends on only the bundle she receives: $u_i(\x) = u_i(x_i)$. In a public decisions instance, agents do not receive separate bundles: instead, the group makes a single decision that affects all agents. We will assume that agents only have utility for their preferred alternative: this will let us standardize notation as follows. For a public decisions outcome $\z$, let $x_{ij}(\z) = z^{j, a_{ij}}$ for all  $i\in N$ and $j \in M$ (we will typically write $x_{ij}(\z) = x_{ij}$ for brevity). Then we can define agent $i$'s \emph{public bundle} as $x_i =(x_{i1}, \dots, x_{im})$. An agent's public bundle represents the fraction of the public decision allocated to her preferred alternative, and so we have $u_i(\z) = u_i(x_i)$ in a public decisions instance as well. 

Throughout the paper, we make the following standard assumptions on each agent's utility function $u_i$:

\begin{enumerate}
	\item Continuous: $u_i: [0,1]^{m} \to \mathbb{R}_{\geq 0}$ is a continuous function.
	\item Normalized: $u_i(0,0,...0) = 0$.
		\item Non-constant: There exists a bundle $x_i$ where $u_i(x_i) > 0$.
	\item Monotone: For any bundles $x_i$ and $x'_i$ where $x_{ij} \geq x'_{ij}$ for all $j$, $u_i(x_i) \geq u_i(x'_i)$.
	\item Concave: For any bundles $x_i$ and $x'_i$ and constant $\lambda \in [0,1]$, we have $u_i(\lambda x_i + (1 - \lambda) x'_i) \geq \lambda u_i( x_i) + (1-\lambda)u_i( x'_i)$.
	\item Homogeneous of degree 1: For any bundles $x_i$ and $x'_i$ and constant $\lambda \geq 0$ where $x_{ij} = \lambda x'_{ij}$ for all $j$, $u_i(x_i) = \lambda u_i(x'_i)$.
\end{enumerate}

The first five are standard assumptions in the market literature. The last is less ubiquitous, but still common: in particular, the vast majority of the popular subclasses of utility functions satisfy this assumption. For example, it is often assumed in real-world applications that utility functions are \emph{linear}, meaning that
\[
u_i(x_i) = \sum_{j \in M} w_{ij} x_{ij}
\]
where $w_{ij} \geq 0$ is the weight agent $i$ has for good $j$. Another important class is \emph{Leontief} functions, where
\[
u_i(x_i) = \min_{j \in M: w_{ij} \ne 0}\ \frac{x_{ij}}{w_{ij}}
\]
Linear utilities imply that goods are independent, whereas Leontief utilities represent perfect complements: goods that only have value in combination. For Leontief utilities, $w_{ij}$ is the relative proportion agent $i$ needs of good $j$.

Both Linear and Leontief utilities are generalized by the class of \emph{constant elasticity of substitution} (CES) utilities, where $u_i(x_i) = \Big(\sum\limits_{j \in M} w_{ij}^{\rho} x_{ij}^{\rho}\Big)^{1/\rho}$ for some constant $\rho \in (-\infty, 0) \cup (0, 1]$. Linear utilities are obtained by setting $\rho = 1$, and taking the limit as $\rho$ approaches $-\infty$ yields Leontief utilities. Taking the limit as $\rho$ approaches $0$ gives Cobb-Douglas utility functions, which have the form $u_i(x_i) = \Big(\prod\limits_{j \in M} x_{ij}^{w_{ij}}\Big)^{1/\sum_{j \in M} w_{ij}}$\footnote{The weights $w_{ij}$ have different interpretations for Leontief utilities vs other CES utilities. For example, if there is only a single good, the CES utility form reduces to $w_{i1} x_{i1}$ and the Leontief utility form reduces to $x_{i1}/w_{i1}$. When we say that taking the limit as $\rho \to - \infty$ yields Leontief utilities, we mean that we obtain the form of Leontief utilities (i.e., a minimization over all the goods).}.

While many of our results hold for any utility functions satisfying our six assumptions, some hold only for particular subclasses; we will make it clear when this is the case.

\subsection{Private goods \& Fisher markets}
\label{sec:model-private}
\label{sec:model-market}

In this work, we primarily consider private goods instances that are Fisher markets. A Fisher market~\cite{brainard_how_2005} is a private goods instance $(N, M)$ where each agent $i$ also has a \emph{budget} $B_i \geq 0$. Each agent's budget can be interpreted as her relative importance. We would typically expect all agents to have the same importance, especially in public decision-making, but allow for the possibility of different budgets for completeness.

For prices $p = (p_1, \dots, p_m) \in \mathbb{R}_{\geq 0}^m$, bundle $x_i$ is \emph{affordable} for agent $i$ if $x_i \cdot p = \sum_{j\in M} x_{ij}p_j \leq B_i$. agent $i$'s \emph{demand set} for prices $p$ is
\[
D_i(p) = \argmax\limits_{x_i \in \bbR^m_{\geq 0}:\ x_i \cdot p \leq B_i} u_i(x_i)
\]
i.e., the set of her favorite affordable bundles. A \emph{market equilibrium} (ME) $(\x, p)$ is an allocation $\x$ and prices $p$ where
\begin{enumerate}
	\item Each agent receives a bundle in her demand set: $x_i \in D_i(p)$.\label{condition:demand}
	\item The market clears: for all $j$, $\sum\limits_{i \in N} x_{ij} \leq 1$. Also, if $p_j > 0$, then $\sum\limits_{i \in N} x_{ij} = 1$.\label{condition:sold-completely}
\end{enumerate}
The most natural case is when all agents have the same budget, in which case the ME is also called the \emph{competitive equilibrium from equal incomes}~\cite{varian_equity_1974}.

Condition~\ref{condition:sold-completely} states that the demand never exceeds the supply, and that any good whose supply is not fully exhausted must have price zero. This implies that agents have no utility for the leftover goods: otherwise they would simply buy more with no additional cost. Note that agents can demand more of a good than the available supply if the cost is less than their budget. It is the role of prices at equilibrium to ensure that demand does not exceed supply. 

Under the first five assumptions on utility functions described in Section~\ref{sec:utilityfuncs}, a market equilibrium is guaranteed to exist for any Fisher market instance~\cite{arrow_existence_1954}. With the addition of the sixth assumption (homogeneity of degree 1), the equilibrium allocations are exactly the allocations maximizing the \emph{Nash welfare}\footnote{This is technically the ``budget-weighted" Nash welfare, but we will omit ``budget-weighted" throughout the paper.}:
\[
NW(\x) = \Big(\prod_{i \in N} u_i(x_i)^{B_i} \Big)^{1/\B}
\]
where  $\B = \sum_{i \in N} B_i$. Maximum Nash welfare allocations can be computed in polynomial time by the celebrated Eisenberg-Gale (EG) convex program~\cite{eisenberg_aggregation_1961, eisenberg_consensus_1959}\footnote{This correspondence still holds under slightly weaker assumptions that our six assumptions~\cite{jain_eisenberggale_2010}.}. The Nash welfare has been lauded as a compromise between fairness and efficiency, and it will be the primary objective function we seek to maximize.


\subsection{Public decisions}\label{sec:model-public}

As in the private markets case, the maximum Nash welfare outcome can be found via a convex program:
\begin{align}
\max\limits_{\z \in [0,1]^{m\times2}} \Big(\prod_{i \in N} u_i(\z)^{B_i} \Big)^{1/\B} \quad 
s.t. \quad z^{j,0} + z^{j,1}\leq 1\ \ \forall j \in M  \label{eqn:egpdm}
\end{align}
 
 The solution to this convex program can be found in polynomial time. This program is very different than the EG program; however, we will show via our reduction that these programs become identical under a transformation of utility functions and issue space.

Furthermore, even without any knowledge of utility functions, a $1/2$ approximation of this program emerges. Because issues are binary, we can very easily guarantee each agent half of her maximum possible utility simply by putting equal probability on each alternative, i.e., $z^{j,0} = z^{j,1} = 1/2$. It follows from concavity and $u_i(0,0,...,0) = 0$ that this also achieves $1/2$ of the maximum possible Nash welfare.

\begin{proposition}\label{prop:half}
Let $\Gamma$ be a public decisions instance $(N,M)$ with agent budgets $B = (B_1...B_n)$, and let $\z$ be the outcome where $z^{j,0} = z^{j,1} = 1/2$ for all $j \in M$. Then $\mfrac{\max\limits_{\mathbf{z'}} NW(\mathbf{z'})}{NW(\z)} \leq 2$\footnote{Whenever we maximize over outcomes of a public decisions instance, i.e., $\max_{\z} NW(\z)$, we implicitly assume that only valid outcomes are considered, meaning that $z^{j,0} + z^{j,1} \leq 1$ for all $j\in M$. The same is true when we maximize over outcomes of a private goods instance, and we adopt these conventions throughout the paper.}.
\end{proposition}

%

In light of this, we would expect any reasonable mechanism for public decision-making to do no worse than this (in terms of Nash welfare), and hopefully do substantially better. Unfortunately, we show in the next section that the natural adaption of Fisher markets to the public decisions setting does no better than this for several important classes of utility functions. Even worse, in the case of linear utilities -- the most important class of utilities in practice -- the Nash welfare can be a factor of $O(n)$ worse than optimal.

\section{Inefficiency of public decisions markets with issue pricing}\label{sec:per-issue-lower}

In a Fisher market, each good is assigned a single price which is common to all agents: thus all agents are treated the same, which is desirable for fairness. This section shows that in the public decisions setting, setting a single price for each issue (issue pricing) can result in very poor equilibria. Although we primarily consider Nash welfare in this paper, the same family of instances will show that the utilitarian welfare (sum of agent utilities) and egalitarian welfare (the minimum agent utility) can also be much worse than optimal.

A \emph{public decisions market} (PDM) consists of a public decisions instance $(N, M)$ along with agent budgets $B = (B_1\dots B_n)$. This definition is independent of the pricing scheme: we  use the term ``PDM" to describe all notions of markets for the public decisions setting. This section uses the following scheme: each issue has a price, each agent uses her budget to purchase probability for her preferred alternatives, and the total probability placed on an alternative is the sum over agents of the probability purchased for that alternative. 

Given per-issue prices $p \in \mathbb{R}_{\geq 0}^m$, a \emph{private bundle} $y_i \in\mathbb{R}_{\geq 0}^m$ is affordable if $y_i \cdot p \leq B_i$. Throughout the paper, we will use $y_i$ to refer to $i$'s private bundle, and $x_i$ to refer to $i$'s public bundle. This distinction only matters in the public decisions setting: we use $y_i$ and $x_i$ interchangeably in the private goods setting.

In this section, for private bundles $\y = (y_1...y_n)$, the corresponding outcome $\z = (z^1...z^m) \in [0,1]^{m\times 2}$ is
\[
z^{j,a} = \sum\limits_{i \in N:\ a_{ij} = a} y_{ij}
\]
The above definition of $\z$ as a function of $\y$ is specific to the issue pricing scheme. The different pricing scheme discussed in Section~\ref{sec:reduction} will define $\z$ differently.

In a Fisher market, an agent's demand set contains the bundles which maximize her utility subject to being affordable. In a PDM with issue pricing, an agent's utility depends not only on her own bundle, but also on other agent's bundles. Thus if we want to define an agent's demand set as the bundles which maximize her utility subject to being affordable, the demand set must depend not only on the prices, but also on the private bundles of other agents. With this in mind, we define the demand set by
\[
D_i(p, y_{-i}) = \argmax\limits_{y_i \in \bbR^m_{\geq 0}:\ y_i \cdot p \leq B_i} u_i(y_{-i}, y_i)
\]
where $y_{-i}$ is the list of private bundles other than that of agent $i$, and with slight abuse of types, $u_i(y_{-i}, y_i)$ is agent $i$'s utility for the outcome when $i$ purchases private bundle $y_i$ and the other agents purchase private bundles $y_{-i}$\footnote{As in the Fisher market setting, agents are allowed to demand more than 1 unit of an issue if the cost is less than their budget. The interpretation of demanding more than unit probability is difficult, but the prices will ensure that this never occurs in equilibrium.}.

An \emph{issue-pricing market equilibrium} (IME) $(\y, p)$ is a list of private bundles $\y$ and issue prices $p \in \mathbb{R}_{\geq 0}^m$ where
\begin{enumerate}
\item Each agent receives a private bundle in her demand set: $y_i \in D_i(p, y_{-i})$.
\item The market clears: for all $j$, $\sum\limits_{i \in N} y_{ij} \leq 1$. Also, if $p_j > 0$, then $\sum\limits_{i \in N} y_{ij} = 1$.\label{condition:sold-public}
\end{enumerate}
By the same reasoning as in the private setting, whenever an issue is not sold completely, agents have no utility for the unsold fraction of the issue\footnote{If some issue $j$ is not sold completely and so $z^{j,0} + z^{j,1} < 1$, one can think of the remaining $1 - z^{j,0} - z^{j,1}$ being allocated to some third option that has no value for any agent.}.

In general it is not known whether every PDM admits an IME. However, for several important utility classes, we give an instance where an IME does exist, but where every IME has poor Nash welfare.

\subsection{Linear utilities}\label{sec:lin-per-issue-lower}
We first show that for linear utilities, an IME always exists. Furthermore, the set of IMEs is identical to the set of private goods MEs that would be obtained if the input were instead treated as a Fisher market (i.e., if each agent's utility only depended on her private bundle).

To see this, we can write agent $i$'s utility for private bundles $\y$ as
\[
u_i(\y) = \sum\limits_{j \in M} w_{ij} \sum\limits_{\substack{k \in N:\\ a_{kj} = a_{ij}}} y_{kj} = 
\sum\limits_{j \in M} w_{ij} y_{ij} + \sum_{j \in M} w_{ij} \sum\limits_{\substack{k \in N\backslash\{i\}:\\ a_{kj} = a_{ij}}} y_{kj}
\]
Agent $i$ cannot affect the actions of other agents, and so has no control over the second term. Thus agent $i$ maximizes her utility by maximizing the first term, $\sum_{j \in M}w_{ij}y_{ij}$, which is exactly the utility function of an agent in a Fisher market. This is expressed formally by Theorem~\ref{thm:lin-per-issue-exist}, whose proof appears in Appendix~\ref{sec:omitted-proofs}.

\begin{restatable}{theorem}{linPerIssueExist}
\label{thm:lin-per-issue-exist}
For a PDM $(N,M, B)$ with linear utilities given by weights $w_{ij} \geq 0$, for every list of private bundles $\y$ and list of prices $p$, $(\y, p)$ is an IME if and only $(\y, p)$ is a ME for the Fisher market $(N,M,B)$ with linear utilities given by the same weights.
\end{restatable}

We now define the family of instances that exhibit poor equilibria in the issue pricing model. For any integer $n\geq 2$ and real number $w \geq 0$, let $\Phi(n, w)$ be the PDM defined by $n = m$, $w_{ii} = w$ for all $i$, $w_{ij} = 1$ for all $j \neq i$, $a_{ii} = 0$, $a_{ij} = 1$ for all $j\neq i$, and $B_i = 1$ for all $i$. In words, on each issue $i$, agent $i$ is alone on one side of the issue, and the other $n-1$ agents are on the opposite side. Each agent $i$ has weight $w$ for issue $i$, and weight $1$ for every other issue.

Our next theorem shows that for linear utilities, the Nash welfare of the IME can be a linear factor worse than optimal. This is especially dreadful in light of how easy it is to achieve half of the optimal Nash welfare via Proposition~\ref{prop:half}.

\begin{restatable}{theorem}{linPerIssueLower}
\label{thm:lin-per-issue-lower}
For any $\epsilon > 0$, $\Phi(n, 1+\epsilon)$ with linear utilities has a unique equilibrium $(\y, p)$, where
\[
\frac{\max\limits_{\mathbf{z'}} NW(\mathbf{z'})}{NW(\y)} \geq \frac{n-1}{1+\epsilon}
\]
\end{restatable}

The proof is in Appendix Section~\ref{sec:omitted-proofs}, but we give some intuition here. We observe that an agent's demand set in a Fisher market always maximizes her ``bang-per-buck" ratio: $w_{ij}/p_j$. To see this, suppose agent $i$ spends some money on a good that does not maximize her bang-per-buck ratio: she could instead spend the same amount of money to get strictly more utility by spending it on a good with maximum bang-per-buck. By Theorem~\ref{thm:lin-per-issue-exist}, this property carries over to the IME.

By symmetry, every issue will have the same price. Since $w_{ii} > w_{ij}$ for all $i$ and for all $j\neq i$, agent $i$'s bang-per-buck ratio is maximized only by good $i$. Thus each agent $i$ spends all of her budget on good $i$. This leads to the outcome where $y_{ii} = 1$ for all $i$, and $y_{ij} = 0$ for all $j \neq i$. Thus $z^{j,0} = 1$ for all $j \in M$. The utility of each agent for this outcome $1+\epsilon$, so the Nash welfare is also $1+\epsilon$. But in the outcome where $z^{j,1} = 1$, for all $j$, each agent has utility $n-1$, so the Nash welfare is $n-1$. This yields the desired bound of $(n-1)/(1+\epsilon)$.

If we used $w_{ii} = w_{ij} = 1$ for all $i,j$, the outcome where $z^{j,0} = 1$ would still be an IME. However, there would now be many more IMEs, including ones with optimal Nash welfare. By setting $w_{ii} = 1+\epsilon$ instead of $w_{ii} = 1$, we can make the outcome where $x_{ii} = 1$ for all $i$ the unique equilibrium. This same issue is not present for Cobb-Douglas and CES utilities with $\rho \in (-\infty, 0) \cup (0, 1)$, which we examine in the next section.

\subsection{Other utilities}\label{sec:other-per-issue-lower}
We briefly mention two results we have for other classes of utility functions. Using the same $\Phi$ construction, Theorems~\ref{thm:cd-per-issue-lower} and \ref{thm:ces-per-issue-lower} state that the Nash welfare of an IME cannot be much better than $1/2$ for Cobb-Douglas utilities and CES utilities, respectively. The formal proofs of Theorems~\ref{thm:cd-per-issue-lower} and \ref{thm:ces-per-issue-lower} appear in Appendix~\ref{sec:omitted-proofs-lower}. 



\begin{restatable}{theorem}{cdLower}
\label{thm:cd-per-issue-lower}
For any IME $(\y, p)$ of $\Phi(n, 1)$ with Cobb-Douglas utilities,
\[
\frac{\max\limits_{\mathbf{z'}} NW(\mathbf{z'})}{NW(\y)} \geq \frac{2 - 2/n}{(n-1)^{1/n}}
\]
\end{restatable}


\begin{restatable}{theorem}{cesLower}
\label{thm:ces-per-issue-lower}
For any IME $(\y, p)$ of $\Phi(n, 1)$ with CES utilities for parameter $\rho \in (-\infty, 0) \cup (0, 1)$, 
\[
\frac{\max\limits_{\mathbf{z'}} NW(\mathbf{z'})}{NW(\y)} \geq 2(1 - 1/n)^{1/\rho}
\]
\end{restatable}

As the number of agents approaches infinity, the bounds in Theorems~\ref{thm:cd-per-issue-lower} and \ref{thm:ces-per-issue-lower} approach $2$. This means that for those classes of utility functions, the issue pricing market model cannot be guaranteed to do better than simply picking the midpoint on every issue (Proposition~\ref{prop:half}). The situation is even worse for linear utilities, where the Nash welfare of an IME can be arbitrarily worse than the optimal Nash welfare. 

One may wonder why Cobb-Douglas and CES utilities with $\rho \in (-\infty, 0) \cup (0, 1)$ do not fail as badly as linear utilities on this family of instances. On a high level, the reason is that both Cobb-Douglas and CES utilities exhibit diminishing returns: the more one buys of a particular good, the less value it adds. This leads to agents splitting their money across multiple goods, regardless of their weights on the individual goods. As a result, small changes in agents' weights end up not affecting their purchases too much. In contrast, for linear utilities, an agent might spend her entire budget on a single good: in fact, if there is a unique good which maximizes her bang-per-buck, she must spend her entire budget on that good. This is exactly the property we use in our inefficiency example, where the fact the $w_{ii} = 1 + \epsilon > w_{ij}$ for $j \neq i$ causes agent $i$ to spend her entire budget on good $i$.

These negative results motivate a more complex market model, which we present in the next section.

\section{Pairwise issue expansion and pairwise pricing}\label{sec:reduction}

In this section, we describe a more complex model of a public decisions market, which relies on \emph{pairwise pricing}: for each issue, there will be a price for each pair of agents who disagree on that issue. We then present our main result: a reduction from any PDM to an equivalent Fisher market. This reduction, which we call \emph{pairwise issue expansion}, can be used to construct a pairwise pricing equilibrium that maximizes the Nash welfare.

The section is organized as follows. Section~\ref{sec:reduction-intro} introduces pairwise issue expansion and gives an informal argument for correctness. Section~\ref{sec:reduction-setup} gives some additional notation and setup, and states our theorems. The formal proofs of correctness are somewhat technical and appear only in the appendix. Finally, Section~\ref{sec:applications} discusses some Fisher market results that this reduction allows us to immediately lift to the public decisions setting.



\subsection{Pairwise issue expansion}\label{sec:reduction-intro}
For any PDM $\Gamma$, we construct a Fisher market $R(\Gamma)$ as follows. The set of agents $N = \{1...n\}$ and their budgets $B_1...B_n$ will be the same. Every issue $j \in M$ will become $O(n^2)$ goods in $R(\Gamma)$. Specifically, for every issue $j$, there will one good for each pair of agents who disagree on issue $j$. Let $R(M)$ be the set of goods in $R(\Gamma)$: then
\[
R(M) = \big\{(i,k,j)\ |\ j \in M,\ i,k \in N,\ a_{ij} \neq a_{kj}\big\}
\]
We we will refer to goods $(k,k',j)$ where $i \in \{k,k'\}$ as agent $i$'s ``pairwise goods". Note that $(i,k,j)$ and $(k,i,j)$ refer to the same good.

If $y_i$ is a bundle associated with $\Gamma$ (denoted $y_i \sim \Gamma$), then $y_i \in \mathbb{R}_{\geq 0}^{|M|}$. If $y_i$ is a bundle associated with $R(\Gamma)$ (denoted $y_i \sim R(\Gamma)$), then $y_i \in \mathbb{R}_{\geq 0}^{|R(M)|}$. 

We will use $j$ to represent issues in $M$ and $\ell$ to represent goods in $R(M)$. We also use $y_{i(ikj)}$ to denote $y_{i\ell}$ when $\ell = (i,k,j)$.

In order to purchase $\alpha$ units of issue $j$ in the PDM, agent $i$ will need to purchase at least $\alpha$ units of all of her pairwise goods for issue $j$. Formally, agent $i$'s utility for a bundle $y_i \in \mathbb{R}_{\geq 0}^{|R(M)|}$ is
\[
u_i\left( \min\limits_{\substack{k \in N:\\ a_{i1} \neq a_{k1}}} y_{i (ik1)},\
\min\limits_{\substack{k \in N:\\ a_{i2} \neq a_{k2}}} y_{i (ik2)},\
\dots\min\limits_{\substack{k \in N:\\ a_{im} \neq a_{km}}} y_{i (ikm)}  \right)
\]
Agent $i$'s utility is as if she purchased $\min\limits_{\substack{k \in N: a_{ij} \neq a_{kj}}} y_{i (ikj)}$ probability of each issue $j$ in the PDM $\Gamma$. For example, if agent $i$'s utility in $\Gamma$ is linear with weights $w_{ij}$, her utility in $R(\Gamma)$ would be
\[
\sum\limits_{j \in M} w_{ij}\Big( \min\limits_{\substack{k \in N:\\ a_{ij} \neq a_{kj}}} y_{i (ikj)}\Big)
\]
These utility functions are \emph{nested Leontief}; this will be discussed formally in Section~\ref{sec:reduction-setup}.

Figure~\ref{fig:reduction} gives a graphical representation of $R(\Gamma)$ for five agents and a single issue $j$, where $a_{1j} = a_{2j} = a_{3j} = 0$ and $a_{4j} = a_{5j} = 1$.  An edge from an agent to a good indicates that that agent desires that good. One key aspect of pairwise issue expansion is that on each issue $j$, each agent is in competition with everyone she disagrees with, and not in competition with anyone she agrees with.

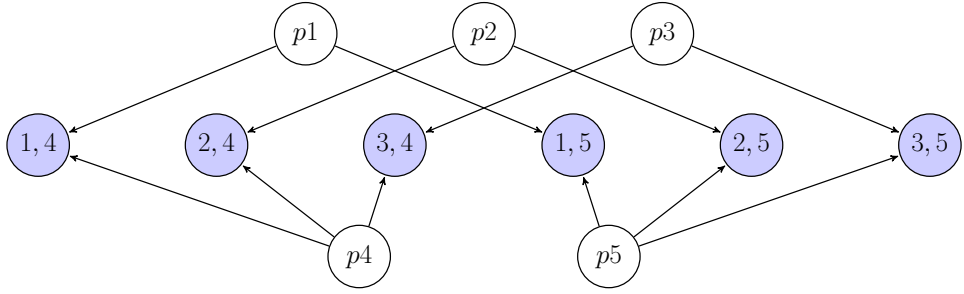
\begin{figure}[h]
\centering
\resizebox{!}{1.5 in}{ 
\begin{tikzpicture}[->,>=stealth',shorten >=1pt,auto,node distance=4cm,
  thick,main node/.style={circle,fill=blue!20,minimum size=14mm,draw,font=\sffamily\LARGE}]
  \newcommand{\hgap}{4}
  \newcommand{\vgap}{2.5}
  \node[main node] (14) at (0,0) {$1,4$};
  \node[main node] (24) at (\hgap, 0) {$2,4$};
  \node[main node] (34) at (2*\hgap, 0) {$3,4$};
  \node[main node] (15) at (3*\hgap, 0) {$1,5$};
  \node[main node] (25) at (4*\hgap, 0) {$2,5$};
  \node[main node] (35) at (5*\hgap, 0) {$3,5$};
   \node [main node, fill=white] (p1) at (1.5*\hgap, \vgap) {$p1$};
   \node [main node, fill=white] (p2) at (2.5*\hgap, \vgap) {$p2$};  
   \node [main node, fill=white] (p3) at (3.5*\hgap, \vgap) {$p3$};  
   \node [main node, fill=white] (p4) at (1.8*\hgap, -\vgap) {$p4$};
   \node [main node, fill=white] (p5) at (3.2*\hgap, -\vgap) {$p5$};
   
  \path[every node/.style={font=\sffamily\normalsize}]
    (p1) edge node {} (14)
    (p1) edge node {} (15)    
    (p2) edge node {} (24)    
    (p2) edge node {} (25)
    (p3) edge node {} (34)        
    (p3) edge node {} (35)   
    (p4) edge node {} (14)     
    (p4) edge node {} (24)    
    (p4) edge node {} (34)    
    (p5) edge node {} (15)    
    (p5) edge node {} (25)    
    (p5) edge node {} (35)    
      ;
\end{tikzpicture}
}
\caption{A graphical representation of the constructed Fisher market $R(\Gamma)$ for five agents and a single issue $j$, where $a_{1j} = a_{2j} = a_{3j} = 0$ and $a_{4j} = a_{5j} = 1$.}\label{fig:reduction}
\end{figure}

We first argue informally for the correctness of the reduction. Agent $i$ will only ever spend money on her pairwise goods, because other goods do not affect her utility. Because of the nested Leontief structure of the utilities in $R(\Gamma)$, for a fixed issue $j$, agent $i$ will buy the same amount of each of her pairwise goods: buying a larger amount of one of the goods would not increase her utility (because it would not increase the minimum), so she would be wasting money. Thus for a fixed issue, agent $i$ buys the same amount of each of her pairwise goods (though this can differ across issues). 

So suppose that for each issue $j$, agent $i$ buys $\alpha_{ij}$ of each of her pairwise goods for that issue. If $R(\Gamma)$ is at equilibrium, every agent $k$ who disagrees with agent $i$ on issue $j$ can receive at most $1-\alpha_{ij}$ of good $(i,k,j)$, since the total supply of each good is $1$. As argued above, agent $k$ will never buy more than $1 - \alpha_{ij}$ of any of her pairwise goods on issue $j$, because of the nested Leontief structure. Thus every agent $k$ who disagrees with agent $i$ on issue $j$ will buy exactly $1- \alpha_{ij}$ of each of her pairwise goods for issue $j$. This leaves exactly $\alpha_{ij}$ for everyone who agrees with agent $i$ on issue $j$. Thus in equilibrium, everyone who agrees with agent $i$ buys $\alpha_{ij}$ of their pairwise goods on issue $j$, and everyone who disagrees with agent $j$ buys $1- \alpha_{ij}$ of their pairwise goods on issue $j$.

This means that when $R(\Gamma)$ is in equilibrium, whenever two agents agree on an issue, they buy the same amount of their pairwise goods for that issue, and whenever they disagree, the amounts they buy sum to 1. Let $\z$ be the outcome where $z^{j,a_{ij}} = \alpha_{ij}$ and $z^{j, 1-a_{ij}} = 1 - \alpha_{ij}$ for all $j \in M$. Then $\z$ is a valid outcome of the PDM. Also, because $R(\Gamma)$ is a Fisher market, an equilibrium price vector assigns a single price to each good: this yields a price for each pairwise disagreement on each issue. This leads to the pairwise pricing equilibrium notion, which $\z$ as defined above will satisfy.

Furthermore, we know that any Fisher market equilibrium maximizes Nash welfare. The agents will have the same utilities in both the PDM and the constructed Fisher market at equilibrium, so the Fisher market equilibrium will respond to a pairwise pricing equilibrium which maximizes Nash welfare in the PDM.

Finally, we mention that this reduction can be generalized to $d$-ary issues under the assumption that each agent has utility for at most one alternative per issue. Instead of one good for each pair of agents who disagree, there would be one good for each set of $d$ agents where each agent has a different preferred alternative, and a similar argument will hold.

\subsection{Additional setup and formal theorem statements}\label{sec:reduction-setup}

Some additional notation will be useful. We define relations $R$ and $\Rfrom$ which will map bundles and prices between $\Gamma$ and $R(\Gamma)$.

For a bundle $y_i \sim\Gamma$, we define a corresponding bundle $R(y_i) \sim R(\Gamma)$ by
\[R(y_i)_{(kk'j)} =
\begin{cases}
y_{ij}\ \text{if}\ i \in \{k,k'\}\\
0\ \text{if}\ i\not\in \{k, k'\}
\end{cases}\hspace{.1 in} \forall (k,k',j) \in R(M)
\]
where $R(y_i)_{(kk'j)}$ denotes the quantity of good $(k, k', j)$ in bundle $R(y_i)$. For a bundle $y_i \sim R(\Gamma)$, the corresponding bundle $\Rfrom(y_i) \sim \Gamma$ is defined by
\[
\Rfrom(y_i)_j = \min\limits_{\substack{k \in N:\\ a_{ij} \neq a_{kj}}} y_{i (ikj)}\hspace{.3 in} \forall j\in M
\]
where $\Rfrom(y_i)_j$ denotes the quantity of issue $j$ in bundle $\Rfrom(y_i)$. Also, for a list of private bundles $\y \sim \Gamma$, we use $R(\y)$ to refer to the list of private bundles in $R(\Gamma)$ where agent $i$'s bundle is $R(y_i)$. Similarly, for any $\y \sim R(\Gamma)$, $\Rfrom(\y)$ is a list of private bundles in $\Gamma$ where agent $i$'s bundle is $\Rfrom(y_i)$. 

It is important to note that while the equilibria of $\Gamma$ and $R(\Gamma)$ coincide, the correspondence is not always meaningful for non-equilibrium outcomes. In particular, not every $y_i \sim R(\Gamma)$ satisfies $y_i = R(\Rfrom(y_i))$: for example if $y_{i(kk'j)} > 0$ when $i \not\in \{k,k'\}$.

Let $u_i$ be agent $i$'s utility function in $\Gamma$. Then agent $i$'s utility function in $R(\Gamma)$ is given by
\[
u_i^R(y_i) = u_i(\Rfrom(y_i))
\]
This is equivalent to the definition of agent utilities given in Section~\ref{sec:reduction-intro}: simply subtitute the definition of $\Rfrom(y_i)$. Also note that for any $y_i \sim \Gamma$, we have $y_i = \Rfrom(R(y_i))$, and so $u_i(y_i) = u_i^R(R(y_i))$.

We would also like to relate prices in $\Gamma$ and $R(\Gamma)$. Since $R(\Gamma)$ is a Fisher market, any price vector $p$ associated with $R(\Gamma)$ (denoted $p \sim R(\Gamma)$) assigns a single price to each good $\ell \in R(M)$: $p \in \mathbb{R}_{\geq 0}^{|R(M)|}$. We will be considering per-person per-issue prices for the PDM $\Gamma$, so any set of prices $p$ associated with $\Gamma$ (denoted $p \sim \Gamma$) assigns one price to each person $i \in N$ for each issue $j \in M$: $p \in \mathbb{R}_{\geq 0}^{m\times n}$.

For a price vector $p \sim R(\Gamma)$, we define prices $\Rfrom(p) \sim \Gamma$ by
\[
\Rfrom(p)_{ij} = \sum\limits_{\substack{k \in N:\\ a_{ij} \neq a_{kj}}} p_{(ikj)}\hspace{.3 in} \forall i\in N, j\in M
\]
where $\Rfrom(p)_{ij}$ is the price of issue $j$ for agent $i$ in price vector $\Rfrom(p)$. In words, $\Rfrom(p)_{ij}$ is the sum of agent $i$'s pairwise prices for issue $j$. We will also use $\Rfrom(p)_i$ to denote the vector of agent $i$'s prices: $\Rfrom(p)_i = (\Rfrom(p)_{i1}...\Rfrom(p)_{im})$.


Before we stating our theorems, we should verify that the utilities in $R(\Gamma)$ satisfy the necessary requirements. If the utility functions in $\Gamma$ are in class $\calh$ ($\calh$ could be the set of linear utility functions, for example), the utility functions in $R(\Gamma)$ will be $\calh$-\emph{nested Leontief}.

\begin{definition}
For some agent $i$, let $f_{i1}, f_{i2}...f_{iL}$ be Leontief utility functions. Then a utility function $u_i$ is $\mathcal{H}$-nested-Leontief if there exists a utility function $h_i: \bbR^L_{\geq 0} \to \bbR_{\geq 0}$ such that $h_i \in \calh$, and
\[
u_i(y_i) = h_i\Big(f_{i1}(y_i), f_{i2}(y_i)\dots f_{iL}(y_i)\Big)
\]
for any bundle $y_i$.
\end{definition}

In our setting, $L = m$ for all agents, and for each $j \in M$, $f_{ij}(y_i) = \Rfrom(y_i)_j = \min\limits_{\substack{k \in N: a_{ij} \neq a_{kj}}} y_{i (ikj)}$. Then for each agent $i$, $u_i^R(y_i) = u_i(\Rfrom(y_i)) = u_i\big(f_{i1}(y_i), f_{i2}(y_i)...f_{iL}(y_i)\big)$.

The next lemma states that as long as $h_i$ and $f_{i1}, f_{i2}...f_{iL}$ satisfy our assumptions on utility functions, their composition does as well.

\begin{restatable}{lemma}{lemNestedUtilities}
\label{lem:nestedhomog}
	Suppose that the functions $h_i, f_{i1}, f_{i2}\dots f_{iL}$ are continuous, normalized, concave, homogeneous of degree 1, non-decreasing, and non-constant. Then $u_i = h_i(f_{i1}, f_{i2}\dots f_{iL})$ meets the same conditions.
\end{restatable}

We will claim that each market equilibrium in $R(\Gamma)$ corresponds to a \emph{pairwise-pricing market equilibrium} (PME) in $\Gamma$. The formal definition of a PME appears in Appendix~\ref{sec:pairwise-pricing}. Informally, it is a list of private bundles $\y$ and per-person per-issue prices $p \in \bbR_{\geq 0}^{m\times n}$ generated by pairwise issue expansion (i.e., $p = \Rfrom(p')$ for some $p' \sim R(\Gamma)$) such that
\begin{enumerate}
\item Every agent receives a private bundle in her demand set.
\item Whenever two agents agree on an issue, they purchase the same amount of that issue.
\item Whenever two agents disagree on an issue, they amounts of that issue that they purchase sum to 1.
\end{enumerate}

This is exactly the definition alluded to via the $\alpha_{ij}$ variables in the informal argument given in Section~\ref{sec:reduction-intro}. This leads to the following theorem:

\begin{restatable}{theorem}{thmReductionEq}
\label{thm:reduction-eq}
For an allocation $\y \sim R(\Gamma)$ and prices $p \sim R(\Gamma)$, $(\y, p)$ is a ME of the market $R(\Gamma)$ if and only if $(\Rfrom(\y), \Rfrom(p))$ is a PME of the PDM $\Gamma$.
\end{restatable}

Finally, we wish to claim the maximum Nash welfare outcomes in $\Gamma$ and $R(\Gamma)$ correspond. We will actually prove this correspondence for all welfare functions, not just the Nash welfare, and even for approximations of welfare functions.

Formally, let $\Psi: \mathbb{R}_{\geq 0}^n \to \mathbb{R}$ be a function. When the $n$ inputs to $\Psi$ are understood to be the $n$ agent utilities for a particular outcome (of a pubic or private instance), we call $\Psi$ a \emph{welfare function}. Because $\Psi$ depends only on the agent utilities, we will use this terminology and notation for both the public and private settings. With slight abuse of types, we will write $\Psi(\z) = \Psi(u_1(\z), u_2(\z),...u_n(\z))$\footnote{Throughout most of the paper, we use $\z$ to refer to the outcome of a public decisions instance and $\x$ to refer to the outcome of a private goods instance. In this discussion, the welfare functions are the same for both public and private instances, so we will use $\z$ to denote outcomes for both.}.

Common welfare functions include the utilitarian welfare function, $\Psi(\z) = \sum_{i \in N} u_i(\z)$, the egalitarian welfare function, $\Psi(\z) = \min_{i \in N} u_i(\z)$, and most importantly for us, the (budget-weighted) Nash welfare function, $\Psi(\z) = \Big(\prod_{i \in N} u_i(\z)^{B_i}\Big)^{1/\B}$. We say that an outcome $\z$ is a $\alpha$-approximation of $\Psi$ if
\[
\Psi(\z) \geq \alpha \cdot \max_{\mathbf{z'} \sim \Gamma} \Psi(\mathbf{z'})
\]


If $\z$ is an outcome of a public decisions instance, technically $R(\z)$ does not typecheck, since $\z = (z^1...z^m)$ is not a list of bundles. We interpret $R(\z)$ to mean $R(x_1, x_2...x_n)$, where $x_i$ is agent $i$'s public bundle as induced by $\z$.

\begin{restatable}{theorem}{thmReductionMax}
\label{thm:reduction-max}
Let $\Psi$ be a welfare function, let $\Gamma$ be the public decisions instance $(N,M)$ with budgets $B_1...B_n$, and let $\alpha \geq 0$. Then $\z$ is an $\alpha$-approximation of $\Psi$ in $\Gamma$ if and only if $R(\z)$ is an $\alpha$-approximation of $\Psi$ in $R(\Gamma)$.
\end{restatable}

Note that by the same reasoning, $\z \sim R(\Gamma)$ is an $\alpha$-approximation of $\Psi$ if and only if $\Rfrom(\z) \sim \Gamma$ is also an $\alpha$-approximation of $\Psi$. Thus for any welfare function $\Psi$ and any $\alpha \geq 0$, the $\alpha$-approximations of $\Gamma$ and $R(\Gamma)$ correspond exactly.

\subsection{Lifting Fisher markets results using pairwise issue expansion}\label{sec:applications}

In addition to uncovering a surprising conceptual connection, pairwise issue expansion allows us to immediately lift many results from the Fisher market literature to the public decision-making setting. In particular, if a result holds for Fisher market with $\calh$-nested Leontief utilities, it holds in the public decisions setting for $\calh$ utilities. Any Fisher market result regarding the ME can be lifted using Theorem~\ref{thm:reduction-eq}, and any Fisher market result regarding any approximation of any welfare function can be lifted using Theorem~\ref{thm:reduction-max}. The following Fisher market results are known:

\begin{enumerate}
\item There exists a strongly polynomial-time algorithm for finding a Fisher market equilibrium with two agents and any utility functions~\cite{chakrabarty_2006_rationality}\footnote{For completeness, we mention an additional mild condition required for this result: the polytope containing the set of feasible utilities of the two agents must be able to be described via a combinatorial LP.}.

\item There exists a strongly polynomial-time algorithm for finding a Fisher market equilibrium for Leontief utilities with weights in $\{0,1\}$ \cite{garg_2005_primal}.

\item There exists a polynomial-time algorithm for a Fisher market with Leontief utilities which yields a $O(\log n)$ approximation simultaneously for all canonical welfare functions (i.e. Nash welfare, utilitarian welfare, egalitarian welfare, etc)~\cite{goel_2014_price}.
\end{enumerate}

The first two can be lifted using Theorem~\ref{thm:reduction-eq}, and the third can be lifted using Theorem~\ref{thm:reduction-max}. Note that nested Leontief-Leontief functions are just Leontief functions. This yields the following PDM results:

\begin{enumerate}

\item There exists a strongly polynomial-time algorithm for finding a PME with two agents and any utility functions.

\item There exists a strongly polynomial-time algorithm for finding a PME for Leontief utilities with weights in $\{0,1\}$.

\item There exists a polynomial-time algorithm for a PDM with Leontief utilities which yields a $O(\log n)$ approximation simultaneously for all canonical welfare functions (i.e. Nash welfare, utilitarian welfare, egalitarian welfare, etc).
\end{enumerate}

The final result we are interested in lifting is a discrete-time \tat process for finding the Fisher market equilibrium for nested CES-Leontief utilities that converges in polynomial-time~\cite{avigdor-elgrabli_convergence_2014}. The next section discusses this in more depth.

As a final comment, there are also results of interest that do not immediately fall under the framework of Theorems~\ref{thm:reduction-eq} and \ref{thm:reduction-max}, but which we conjecture could be lifted using our reduction. For example, \cite{goel_2018_beyond} studies Leontief utilities in the context of \emph{price curves}, where the cost of a good may be any increasing function of the quantity purchased (as opposed to just a linear function, as is usually assumed). We believe that pairwise issue expansion could be used to generate \emph{pairwise price curves}, where there would be a price curve assigned to each pair of agents who disagree on an issue, instead of a single price. Pairwise issue expansion seems general enough to apply to non-standard market models, such as that of \cite{goel_2018_beyond}, but we leave it for future work.

\section{Public market t\^{a}tonnements}
\label{sec:tat}
\begin{figure}[ht!bp]
	\centering
	\resizebox{!}{1.35 in}{ 
		\begin{tikzpicture}[->,>=stealth',shorten >=1pt,node distance=5.5cm]
		\tikzstyle{every state}=[align=center]
		\node[state,minimum size=80pt] (Start)
		{Public\\market \\agents};
		\node[state,minimum size=80pt] (Mitte) [right of=Start] 
		{Public to\\private \\reduction};
		\node[state,minimum size=80pt] (Ende)   [right of=Mitte] 
		{Private\\market\\t\^{a}tonnement};
		\path (Start) edge [bend left]               node[above] {demands $y$ at time $t$} (Mitte);
		\path (Mitte) edge [bend left]              node[below] {prices $\Rfrom(p^{t+1})$} (Start);
		\path (Mitte) edge [bend left]               node[above] {demands $R(y)$ at time $t$} (Ende);
		\path (Ende) edge [bend left]              node[below] {prices ${p}^{t+1}$} (Mitte);
		\end{tikzpicture}
	}
	\caption{$\Rfrom(\mathcal{T})$ illustration using a hidden private market \tat}\label{fig:tatonnement}
\end{figure}

In this section, we describe how the reduction immediately leads to existence of several public market t\^{a}tonnements. In particular, we show that any Fisher market \tat that works for $\calh$-nested utilities yields a PDM \tat for $\calh$ utilities. 

This does not immediately follow from pairwise issue expansion for several reasons. The first is that \tat deals with approximate equilibria, and Theorem~\ref{thm:reduction-eq} only considers exact equilibria. Because this correspondence holds for approximate equilibria as well (see proof of Theorem 5.4), any Fisher market \tat that works for $\calh$-nested utilities immediately yields an \emph{algorithm} for computing PDM equilibria for $\calh$ utilities, but not a \tat. A true PDM \tat would only have access to agents' demands in the PDM, but the resulting algorithm would need to elicit agents' demands in the constructed Fisher market. We handle this by running the Fisher market \tat as a hidden subroutine within the PDM \tat, as demonstrated by Figure~\ref{fig:tatonnement}.

Let a Fisher market \tat $\mathcal{T}$ be an iterative algorithm that starts at an initial price vector $p^0$, and then at each time $t$, \begin{enumerate}
	\item Receives demand set $D_i(p^t)$\footnote{With strictly concave utility functions, each agent's demand at a given price is unique. With linear utilities, the demand set can be expressed as the set of goods that are equally desirable at the given prices. Also, we assume that agents are \emph{price-taking}, meaning that they honestly report their demand given a set of prices, and do not anticipate how prices will change as a result of their actions.} from each agent $i$.
	\item Updates prices as some function $g$ of the demands, $p^{t+1} = g_{\mathcal{T}}(p^t, D(p^t))$. 
\end{enumerate}
As time increases, prices and associated demands approach an approximate equilibrium, for some notion of approximate. Figure~\ref{fig:tatonnement} illustrates the meta-algorithm for public market t\^{a}tonnements. From a Fisher market \tat $\mathcal{T}$, let $\Rfrom(\mathcal{T})$ be the induced public market t\^{a}tonnement that initializes an initial price vector $p^0$ in the hidden Fisher market and then at each time $t$,\begin{enumerate}
	\item Converts prices $p^t$ to public market prices $\Rfrom(p^{t})$ and shows them to agents.
	\item Receives demand set $D_i(\Rfrom(p^{t}))$ from each agent $i$.
	\item Converts the agent demands to the associated private market demand set $\mathbf{y}^R = \{R(y_i)\}_{y_i\in D_i(\Rfrom(p^{t}))}$. 
	\item Updates prices through the Fisher market \tat function, $p^{t+1} = g_{\T}(p^t, \mathbf{y}^R)$. 
\end{enumerate}


We begin with the definitions of approximate equilibria and convergence.

\begin{definition}
	A $\delta$-equilibrium $(\x, p)$ in a Fisher market  is an allocation $\x$ and prices $p$ where
	\begin{enumerate}
		\item Each agent receives a bundle in her demand set: $x_i \in D_i(p)$.\label{condition:appdemand}
		\item $p_j > \delta \implies \sum\limits_{i \in N} x_{ij} > 1-\delta$ \label{con:fisherdelp} 
		\item $\forall j$, $\sum\limits_{i \in N} x_{ij} \leq 1+\delta$
	\end{enumerate}
\label{def:fishapprox}
\end{definition}

Note that this definition is introduced in~\cite{avigdor-elgrabli_convergence_2014}.

\begin{definition}
A $\delta$-PME $(\y, p)$ is a list of private bundles $\y$ and per-person per-issue prices $p \in \mathbb{R}_{\geq 0}^{m\times n}$ where
\begin{enumerate}
	\item Each agent receives a private bundle in her demand set: $y_i \in D_i(p_i)$.
	\item The market approximately clears: there exists a outcome $\z = (z^1...z^m) \in [0,1]^{m\times 2}$ where for every issue $j \in M$, all of the following hold:
	\begin{enumerate}
		\item $z^{j,0} + z^{j,1} \leq 1+\delta$
		\item For all $i\in N$, $y_{ij} \leq z^{j, a_{ij}} + \delta$. If $p_{ij} > n\delta$, then $y_{ij} > z^{j,a_{ij}} - \delta$.
	\end{enumerate}\label{condition:market-clears}
\end{enumerate}
\label{def:pdmapprox}
\end{definition}

\begin{definition}
	A \tat $\T$ has converged to a $\delta$-equilibrium at time $T$ if $\exists \y$ where $(\y, p^T)$ is a $\delta$-equilibrium. Similarly, $\Rfrom(\T)$ has converged to a $\delta$-PME at time $T$ if $\exists \y$ where $(\y, \Rfrom(p^T))$ is a $\delta$-PME.
	\end{definition}

The definition does not imply that \textit{all} demands at the equilibrium prices form an approximate equilibrium, only that there exists an allocation consistent with demands at the equilibrium prices such that the supply constraints are met. However, note that when utility functions are strictly concave, demands are unique.

Our first theorem shows that any Fisher market \tat for $\mathcal{H}$-nested Leontief utility functions yields a PDM \tat for $\mathcal{H}$ utility functions. Theorem $\ref{thm:tatonnementlifting}$, whose proof appears in the appendix, allows one to lift both convergence and convergence rates from Fisher market t\^{a}tonnements. 

\begin{restatable}{theorem}{tatLift}
	Consider a Fisher market t\^{a}tonnement $\mathcal{T}$. Suppose $\mathcal{T}$ converges to a $\delta$-equilibrium for $\mathcal{H}$-nested leontief utilities in $O(\kappa(m,n,\delta))$ time steps, where $n$ is the number of agents and $m$ the number of goods. Then $\Rfrom(\mathcal{T})$ converges to a $3\delta$-PME for the PDM with $\mathcal{H}$ utilities in $O(\kappa(n^2m,n,\delta))$. 
	\label{thm:tatonnementlifting}
\end{restatable}

%
%
%

One Fisher market \tat that we can lift using Theorem~\ref{thm:tatonnementlifting} comes from \cite{avigdor-elgrabli_convergence_2014}, which gives a polynomial-time t\^{a}tonnement that converges to a $\delta$-equilibrium for CES-Leontief utilities with $\rho \in (-\infty, 0)\cup(0, 1)$ in polynomial time. By Theorem~\ref{thm:tatonnementlifting}, this yields a PDM \tat for CES utilities.


We would also like a PDM \tat that works for a wider range of utility functions, especially linear utilities. Section~\ref{sec:tatprivatedetour} presents a stochastic gradient descent style t\^{a}tonnement for Fisher markets which converges asymptotically to an equilibrium for all EG utility functions\footnote{functions that meet the 5 conditions from Section~\ref{sec:model}.}, following the framework of~\cite{cheung_tatonnement_2013}. Combined with Theorem~\ref{thm:tatonnementlifting}, this t\^{a}tonnement implies existence of a PDM t\^{a}tonnement with asymptotic convergence to an equilibrium for all EG utility functions.






\section{Conclusion}

In this paper, we studied adaptations of markets to the public decision-making setting. In Section~\ref{sec:per-issue-lower}, we showed that issue pricing in the public decisions setting can yield very poor equilibria: for linear utilities, the Nash welfare can be a factor of $O(n)$ worse than optimal. This is in contrast to private goods, where per-good pricing is the accepted standard, and yields optimal equilibria. We showed in Section~\ref{sec:reduction} that pairwise issue expansion reduces any public decisions market to an equivalent Fisher market, how optimal equilibria can be constructed using this reduction. We used pairwise issue expansion to lift various Fisher market results to the public decision-making context, including \tat, which we discussed in Section~\ref{sec:tat}.

Most importantly, our reduction uncovers a powerful connection between the private goods and public decision-making settings that we believe has many possible applications. For example, suppose we had a mechanism for private goods which computes some desirable outcome other than maximum Nash welfare (maybe it computes the allocation which maximizes the minimum utility, for example). If that algorithm works for nested $\calh$-Leontief utilities for private goods, we imagine that it could be immediately lifted to work for $\calh$ utilities in for public-decisions. More generally, it seems like more or less any result that applies to $\calh$-Leontief utilities for private goods would apply for $\calh$ utilities for public decisions. We believe this merits more study.

\section*{Acknowledgements}
This research was supported in part by NSF grant CCF-1637418, 
ONR grant N00014-15-1-2786, and the NSF Graduate Research Fellowship under grants DGE-114747 and DGE-1656518.

\bibliographystyle{plain}
\bibliography{bibliography}

\appendix

\section{Lindahl Equilibria}\label{sec:foley}
In this section, we show how our public decisions setting corresponds to a natural public goods market in the setting of Lindahl equilibria, and how our reduction can also be used to compute Lindahl prices for this public goods market. The Lindahl Equilibrium has a long history and, at times, the term has been used to mean slightly different things~\cite{van_den_nouweland_lindahl_2015}. Foley \cite{foley_lindahls_1970} gave general conditions for the existence of the Lindahl Equilibrium and its correspondence to the core. We first introduce a simplified definition of Lindahl Equilibria. 

\begin{definition}[\cite{foley_lindahls_1970, van_den_nouweland_lindahl_2015}]
	A Lindahl Equilibrium with $m$ public goods, 1 private good, $n$ agents, and entry private good amounts $\{w_i\}_{i=1}^n$ is an (public goods allocation, private goods allocation, per-person per-issue price) vector $(z^*\in \bbR_+^m, \{y_i^* \in \bbR_+\}_{i=1}^n, \{p_i^* \in \bbR_+^m\}_{i=1}^n)$ such that
	\begin{enumerate}
		\item $z^*$ is a solution to $\max_z \left(\sum_{j=1}^n p_j^*\right)z - c(z)$
		\item  For each $i\in \{1, \dots, n\}$, $(z^*, y_i^*)$ is a solution to	$\max_{(z,y_i)}
		u_i(z, y_i)$ subject to $p^*_i z + y_i \leq w_i$
	\end{enumerate}
where $c(z)$ is the cost to produce the public good vector $z$ in terms of the private good and $u_i$ is the participant utility function in terms of the public and private goods. \label{def:lindahl}
\end{definition}

In a Public Decision Market, the private good is the ``influence'' of each agent, for which agents have no utility, i.e., influence not spent is lost. Furthermore, there are $2$ public goods per issue, 1 for each alternative, and each with the same price for each agent. Similarly,
\[
c(z) = \begin{cases*}
0 & $z^{j,0} + z^{j,1} \leq 1\,\,\,\forall j$\\
\infty & \text{else}
\end{cases*}
\]
i.e., if in the case in which each alternative on each issue is implemented with some probability, then there is no cost of using the entire probability, and no possibility of creating more probability.

\begin{lemma}
	An equilibrium $(z^*\in \bbR_+^{2\times m}, \{y_i^* \in \bbR_+\}_{i=1}^n, \{p_i^* \in \bbR_+^{2\times m}\}_{i=1}^n)$ of the PDM with $m$ issues and agent budgets $B_i$, where $z^*$ is the decision vector, $p^*$ are the per-person per-issue prices from the Fisher market reduction, and final influence vectors are $y_i^* = 0$, is a Lindahl Equilibrium with entry private good amounts $w_i = B_i$. \label{lem:pdmlindahl}
\end{lemma}
\begin{proof}
	Condition 2 follows directly from the equilibrium condition that optimal allocation is in the demand set of each agent at the equilibrium prices. Condition 1 requires a bit more work. Note that $\forall j$, for each copy of the good, the sub-price $p_{ikj} = p_{kij}$, where $a_{ij} \neq a_{kj}$. Then, $\forall j, p^0_j = \sum_{i:a_{ij} = a} \sum_{k:a_{kj} \neq a} p_{ikj}$ $\implies$ $p^0_j = p^1_j$. Thus, all feasible $x$ are in the solution set in the first condition, and an equilibrium of the PDM is feasible.
\end{proof}

Lemma~\ref{lem:pdmlindahl}, alongside Theorem~\ref{thm:reduction-eq} and the existence of Fisher market equilibria~\cite{arrow_existence_1954} establishes the existence of a Lindahl equilibrium in our setting. 

We note that Lemma~\ref{lem:pdmlindahl} further establishes that the solution is in the \textit{core}, as Lindahl Equilibria are in the core~\cite{foley_lindahls_1970}.
\subsection{Economies with public goods}
The existence of Lindahl Equilibria in Public Good economies (of which Public Decision Markets are a special case, as we will show) was established by Foley~\cite{foley_lindahls_1970}. The chief technique is a reduction to Private Goods economies. The reduction yields a non-constructive existence proof, and operates as follows: create a copy of each good for each participant, with equality of the amount of each good enforced through conic constraints. Then, the proof is finished by invoking the existence of equilibria satisfying certain conditions in Private Goods economies, after showing that the additional constraints restricting the cone of production do not violate any assumptions~\cite{debreu_theory_1959}. We note that existence of Lindahl equilibria of the PDM can also be established non-constructively through the same technique, by showing that the resulting market satisfies the assumptions in \cite{foley_lindahls_1970}. 

One natural question is how Foley's reduction to private goods compares to the reduction in this work. Equation~\eqref{eqn:egprogramfoley} contains the convex program to find MNW through Foley's reduction. We use $s\in \{0,1\}$ to denote each side of the issue. Equation~\eqref{eqn:egprogram2} contains our reduction, which has a nested utility function structure. 


\begin{align}
\max\limits_{\x \in [0,1]^{(2\times m)\times n}} &\Big(\prod_{i \in N} u_i(x_i)^{B_i} \Big)^{1/\B}\nonumber\\
s.t.\,\,\,\,&x^0_{ij} + x^1_{ij}\leq 1&\forall j \in M,\ i \in N \label{eqn:egprogramfoley}\\
&x^s_{ij} = x^s_{kj}&\forall j \in M, \ i,k \in N, s\in\{0, 1\} \nonumber\\\nonumber\\
\max\limits_{v \in [0,1]^{m\times n}, \x \in [0,1]^{\tilde{n}}} &\Big(\prod_{i \in N} u_i(v_i)^{B_i} \Big)^{1/\B}\nonumber\\
s.t.\,\,\,\,&v_{ij} \leq x_{i(ikj)}& \forall j \in M,\ i,k \in N, a_{ij}\neq a_{kj} \label{eqn:egprogram2} \\
&x_{i(ikj)} + x_{k(ikj)} \leq 1& \forall j \in M,\ i,k \in N, a_{ij}\neq a_{kj}\nonumber 
\end{align}
Where $\tilde{n}$ is $2(\text{\# of disagreement pairs across issues})$. 
Both programs can be solved in polynomial time. However, Foley's reduction does not obviously resemble a Fisher market due to the extra equality constraints. 
\begin{thmComment}
	Program~\eqref{eqn:egprogramfoley} does not transform into a Fisher market.
\end{thmComment}
\begin{proof}
	We write the Lagrangian of the Program~\eqref{eqn:egprogramfoley}, with the objective function written in log form.
	\begin{align*}
	L(x, p, q) &= \sum_i B_i \log (u_i(x_i)) - \sum_{i,j} p_{ij}(x^0_{ij} + x^1_{ij}) + \sum_{i,j}p_{ij} - \sum_{i\neq k,j,s} q^s_{ikj} (x^s_{ij} - x^s_{kj})\\ & s.t.\,\,\,\, p\geq 0, x\geq0, q^s_{ikj}\geq 0\\
	\end{align*}
	This Lagrangian has per-person per-issue prices $q^s_{ikj}, p_{ij}$ for each side of each issue that cannot be trivially turned into per-good prices with separate goods not having joint constraints. If one considers each $(i,j,s)$ tuple a good in the Fisher market, the two goods associated with the two sides of each issue, $x^0_{ij}, x^1_{ij}$ are coupled through $p_{ij}$. Similarly, if pair $(i,j)$ corresponds to a good, goods across individuals $x^s_{ij}, x^s_{kj}$ are coupled through $q^s_{ikj}$ in a way that does not resemble supply constraints. Merging the goods for each side, such that the good represents the probability on one of the alternatives for each issue, would violate the non-decreasing constraint for utility functions. Eliminating or combining either of these variables would amount to eliminating the corresponding constraints. This coupling prevents a non-trivial mapping to the Fisher market case, which does not have cross-good constraints but has per-good prices. 
\end{proof}

 Our reduction fills this gap for Public Decision markets, showing that these prices can emerge from a pure Fisher market with a modification of buyer utilities. In Program~\eqref{eqn:egprogram2}, there is a good $(i,k,j)$ for each $(i,k)$ pair that disagrees on issue $j$, with $i$'s utility function only dependent on the amount $i$ buys, $x_{ikj}$. This program thus yields goods with per-good pricing (on the Fisher market goods) and no cross-goods constraints.

\begin{proposition}
Programs~\eqref{eqn:egprogramfoley}~and~\eqref{eqn:egprogram2} are equivalent.
\end{proposition}
\begin{proof}
	Both are equivalent to Program \eqref{eqn:egpdm}, which we repeat below:\begin{align}
	\max\limits_{\z \in [0,1]^{2\times m}} &\Big(\prod_{i \in N} u_i(z)^{B_i} \Big)^{1/\B}\nonumber\\
	s.t.\,\,\,\,&z^{j,0} + z^{j,1}\leq 1\ \ &\forall j \in M\nonumber 
	\end{align}
Program~\eqref{eqn:egprogramfoley} and~\eqref{eqn:egpdm} are immediately equivalent by combining variables. The equivalence of the reduction (Theorem~\ref{thm:reduction-eq}) establishes that Program~\eqref{eqn:egprogram2} and PDM Program~\eqref{eqn:egpdm} have the same solution. 
\end{proof}

\section{Omitted definitions and proofs from Section 4}

Section~\ref{sec:pairwise-pricing} contains the formal definitions of the pairwise pricing model. Section~\ref{sec:reduction-proofs} contains the formal analysis of $R$ and $\Rfrom$, leading to Theorems~\ref{thm:reduction-eq} and \ref{thm:reduction-max}. 

\subsection{Pairwise pricing}\label{sec:pairwise-pricing}

In the issue pricing model of Section~\ref{sec:per-issue-lower}, the price for an issue was the same for all agents, and the amount of probability put on alternative $a$ on issue $j$ in the outcome (denoted $z^{j,a}$) was the sum of the agents' purchases. Formally, Section~\ref{sec:per-issue-lower} defined $z^{j,a} = \sum\limits_{k \in N:\ a_{ij} = a} y_{ij}$ where $y_{ij}$ is the probability that agent $i$ purchased on issue $j$ ($y_i$ is agent $i$'s private bundle).

The definition of $z^{j,a}$ will be different here. This section describes a model where agents may have different prices for the same issue. This will allow us to enforce that in equilibrium, all agents who agree on issue $j$ will purchase the same amount of issue $j$. For every $i \in N$ and $j \in M$ where agent $i$'s price for issue $j$ is nonzero, at equilibrium\footnote{The outcome will not be well-defined for a list of private bundles not at equilibrium, since agents may have incompatible demands. This will not be important; we mention it only for completeness.} we will have
\[
z^{j, a_{ij}} = y_{ij}
\] 
The key consequence is that each agent's private and public bundles will be the same in equilibrium, and so $u_i(y_i) = u_i(x_i(\z))$. Thus each agent's utility can be written as a function of only her private bundle: this will enable the reduction to private goods.

We now formally describe the personalized pricing model. For prices $p \in \bbR_{\geq 0}^{m\times n}$, let $p_{ij}$ be the price for agent $i$ for issue $j$, and let $p_i = (p_{i1}...p_{im})$. Formally, a private bundle $y_i$ is affordable if $y_i \cdot p_i = \sum_{j\in M}y_{ij}p_{ij} \leq B_i$. Because we will have $u_i(y_i) = u_i(x_i)$ in equilibrium, we can define agent $i$'s demand to be independent of other agents' private bundles\footnote{We assume that agents truthfully report their demands according to this definition: recall that we do not consider strategic behavior.}:
\[
D_i(p) = \argmax\limits_{y_i \in \bbR^m_{\geq 0}:\ y_i \cdot p_i \le B_i} u_i(y_i)
\]

A \emph{personalized-pricing market equilibrium} (PME) $(\y, p)$ is a list of private bundles $\y$ and personalized prices $p \in \mathbb{R}_{\geq 0}^{m\times n}$ where
\begin{enumerate}
\item Each agent receives a private bundle in her demand set: $y_i \in D_i(p_i)$.
\item The market clears: there exists an outcome $\z = (z^1...z^m) \in [0,1]^m$ where for every issue $j \in M$, all of the following hold:
	\begin{enumerate}
	\item $z^{j,0} + z^{j,1} = 1$
	\item For all $i\in N$, $y_{ij} \leq z^{j, a_{ij}}$. If $p_{ij} > 0$, then $y_{ij} = z^{j,a_{ij}}$.
	\end{enumerate}\label{condition:market-clears}
\end{enumerate}

The market clearing condition (Condition~\ref{condition:market-clears}) is different than in traditional private goods markets. Instead of the sum of the agent's demands being equal to the supply, the condition here is that there is a single outcome that is consistent with every agent's demand. Roughly speaking, this means that whenever two agents agree on an issue, they demand the same quantity of that issue, and whenever two agents disagree, the sum of their demands equals the supply. This can be thought of as all agents buying the ``same" private bundle, modulo their preferred alternatives. 

At equilibrium, $\z$ is treated as the outcome of the public decisions instance. However, $\z$ may not be unique: if $y_{ij} < z^{j, a_{ij}}$ for some $i,j$, there may be multiple outcomes compatible with the list of agent demands. The following proposition shows that all outcomes compatible with $\y$ are more or less the same.


\begin{proposition}\label{prop:same-utility}
Let $(\y, p)$ be a PME. Then for any outcome $\z$ satisfying the market clearing condition, $u_i(\z) = u_i(y_i)$ for all $i\in N$.
\end{proposition}
\begin{proof}
Fix some agent $i \in N$, and let $y'_i$ be the private bundle where $y'_{ij} = z^{j, a_{ij}}$ for all $j \in M$. For every issue $j$ where $y_{ij} \neq z^{j, a_{ij}}$, we have $p_{ij} = 0$. Thus $y_i$ and $y'_i$ have the same cost. Since $y_i$ is in agent $i$'s demand set, $y_i$ is affordable. Thus $y'_i$ is also affordable. Suppose $u_i(\z) = u_i(y'_i) > u_i(y_i)$: then $y_i$ would not be in agent $i$'s demand set, which is a contradiction.
\end{proof}

Since each agent's private and public bundles are the same at equilibrium in this model, we mostly omit ``private" and ``public" and just use the term ``bundle". We reserve $\z$ for denoting the overall outcome of the PDM, and just use $y_i$ to denote agent $i$'s bundle.

\subsection{Formal analysis of pairwise issue expansion}\label{sec:reduction-proofs}

We begin with Lemma~\ref{lem:reduction-cost}, which states that as long as agents only buy their pairwise goods, the cost of a bundle $y_i \sim R(\Gamma)$ at prices $p$ is the same as the cost of bundle $\Rfrom(y_i)\sim \Gamma$ at prices $\Rfrom(p)$. The proof primary consists of arithmetic and substituting definitions.

\begin{lemma}\label{lem:reduction-cost}
Given prices $p \sim R(\Gamma)$ and a bundle $y_i \sim R(\Gamma)$,
\begin{enumerate}
\item $\Rfrom(y_i) \cdot \Rfrom(p)_i \leq y_i \cdot p$
\item Suppose that (1) for any $j \in M$ and $k,k' \in N\backslash \{i\}$ where $y_{i(kk'j)} \neq 0$, we have $p_{(kk'j)} = 0$, and (2) for any $j \in M$ and $k\in N$ where $y_{i(ikj)} \neq \min\limits_{\substack{k \in N:\\ a_{ij} \neq a_{kj}}} y_{i (ikj)}$, we have $p_{(ikj)} = 0$. Then $\Rfrom(y_i) \cdot \Rfrom(p)_i = y_i \cdot p$.
\end{enumerate}
\end{lemma}

\begin{proof}
Suppose $y_i$ is a bundle in $R(\Gamma)$. The cost of $y_i$ at prices $p$ is 
\begin{align*}
y_i \cdot p =&\ \sum\limits_{\ell \in R(M)} y_{i\ell}p_\ell &\ \text{(by definition)}\\
=&\ \sum\limits_{j \in M} \sum\limits_{\substack{k,k' \in N:\\ a_{kj} \neq a_{k'j}}} y_{i(kk'j)}p_{(kk'j)}&\ \text{(rewriting each good $\ell \in R(M)$ as a triple $(i,k,j)$)}\\
\geq&\ \sum\limits_{j \in M} \sum\limits_{\substack{k \in N:\\ a_{ij} \neq a_{kj}}} y_{i(ikj)}p_{(ikj)} &\ \text{(only including agent $i$'s pairwise goods in the sum)}\\
\geq&\ \sum\limits_{j \in M} \sum\limits_{\substack{k \in N:\\ a_{ij} \neq a_{kj}}} p_{(ikj)}
\min\limits_{\substack{k' \in N:\\ a_{ij} \neq a_{k'j}}} y_{i (ik'j)}&\ \text{(replacing each $y_{i(ikj)}$ with $\min\limits_{\substack{k' \in N:\\ a_{ij} \neq a_{k'j}}} y_{i (ik'j)}$)}\\
=&\ \sum\limits_{j \in M} \sum\limits_{\substack{k \in N:\\ a_{ij} \neq a_{kj}}} p_{(ikj)}\Rfrom(y_i)_j &\ \text{(by definition)}\\
=&\ \sum\limits_{j \in M}\Rfrom(y_i)_j \sum\limits_{\substack{k \in N:\\ a_{ij} \neq a_{kj}}} p_{(ikj)} &\ \text{(moving $\Rfrom(y_i)_j$ out of the inner sum)}\\
=&\ \sum\limits_{j \in M}\Rfrom(y_i)_j \Rfrom(p)_{ij} &\ \text{(by definition)}\\
=&\ \Rfrom(y_i) \cdot \Rfrom(p)_i &\ \text{(by definition)}
\end{align*}
Furthermore, the first inequality holds with equality if for any $j \in M$ and $k,k' \in N\backslash \{i\}$ where $y_{i(kk'j)} \neq 0$, $p_{(kk'j)} = 0$. Similarly, the second inequality holds with equality if for any $j \in M$ and $k\in N$ where $y_{i(ikj)} \neq \min\limits_{\substack{k \in N:\\ a_{ij} \neq a_{kj}}} y_{i (ikj)}$, $p_{(ikj)} = 0$. Therefore under those two assumptions, $\Rfrom(y_i) \cdot \Rfrom(p)_i = y_i \cdot p$.
\end{proof}

Lemma~\ref{lem:reduction-forward-cost} is a simple application of Lemma~\ref{lem:reduction-cost}.
\begin{lemma}\label{lem:reduction-forward-cost}
For prices $p \sim R(\Gamma)$ and a bundle $y_i \sim \Gamma$, we have $y_i \cdot \Rfrom(p)_i = R(y_i) \cdot p$.
\end{lemma}

\begin{proof}
By definition of $R(y_i)$, we have (1) $y_{i(kk'j)} = 0$ for all $j \in M$ and $k,k' \in N\backslash \{i\}$, and (2) $y_{i(ikj)} = \min\limits_{\substack{k \in N:\\ a_{ij} \neq a_{kj}}} y_{i(ikj)}$ for all $j \in M$ and $k\in N$. Then by Lemma~\ref{lem:reduction-cost}, $y_i \cdot \Rfrom(p)_i = R(y_i) \cdot p$.
\end{proof}

Lemma~\ref{lem:reduction-demand-pairwise} states that if a bundle $y_i \sim R(\Gamma)$ is agent $i$'s demand set $D_i^R(p)$, then (1) $y_i$ contains only agent $i$'s pairwise goods, and (2) for a fixed issue $j$, $y_i$ contains the same amount of each of her pairwise goods. The proof is based on the informal argument given before: violating either (1) or (2) wastes money that could be spend to increase her utility.

\begin{lemma}\label{lem:reduction-demand-pairwise}
Given prices $p\sim R(\Gamma)$, suppose a bundle $y_i  \sim R(\Gamma)$ is in $D_i^R(p)$. Then (1) for any $j \in M$ and $k,k' \in N\backslash \{i\}$ where $y_{i(kk'j)} \neq 0$, we have $p_{(kk'j)} = 0$, and (2) for any $j \in M$ and $k\in N$ where $y_{i(ikj)} \neq \min\limits_{\substack{k \in N:\\ a_{ij} \neq a_{kj}}} y_{i (ikj)}$, we have $p_{(ikj)} = 0$.
\end{lemma}

\begin{proof}
First, suppose for sake of contradiction that there exists $j \in M$ and $k, k' \in N\backslash\{i\}$ where $p_{(kk'j)} > 0$ and $y_{i(kk'j)} > 0$. Consider the bundle $y_i' \sim R(\Gamma)$ which is identical to $y_i$, except that $y'_{i(kk'j)} = 0$. Since $\Rfrom(y_i) = \Rfrom(y_i')$, we have $u_i^R(y_i) = u_i^R(y_i')$. But since $p_j > 0$, the $y_i \cdot p - y_i'\cdot p = y_{i(kk'j)}p_{(kk'j)}$. Consider the bundle $y''_i \sim R(\Gamma)$ where for all $\ell \in R(M)$,
\[
y''_{i\ell} = y'_{i\ell} + \frac{y_{i(kk'j)}p_{(kk'j)}}{\sum_{\ell' \in R(M)} p_{\ell '}}
\]
Then we have $y''_i \cdot p = y_i' \cdot p + y_{i(kk'j)}p_{(kk'j)} = y_i \cdot p$. Since $y_i \in D_i^R(p)$, $y_i$ is affordable at prices $p$. Thus $y''_i$ is affordable at prices $p$. 

Finally, we show that $u_i^R(y''_i) > u_i^R(y_i)$. We have $y''_{i\ell} > y'_{i\ell}$ for all $\ell \in R(M)$. Thus for all $\ell$, there exists a constant $\alpha_\ell > 1$ where $y''_{i\ell} = \alpha_\ell y'_{i\ell}$. Let $\alpha = \min_{\ell \in R(M)} \alpha_\ell$. Then $y''_{i\ell} \geq \alpha y'_{i\ell}$ for all $\ell \in R(M)$. Because $u_i^R$ is homogenous of degree 1 and monotone, we have $u_i^R(y''_{i\ell}) \geq \alpha \cdot u_i^R(y'_{i\ell}) > u_i^R(y'_{i\ell}) = u_i^R(y_i)$.

Thus we have $u_i^R(y''_{i\ell}) > u_i^R(y_{i\ell})$ and $y''_i \cdot p = y_i' \cdot p$. But then $y_i$ cannot be in agent $i$'s demand set, which is a contradiction.

The second case is similar. Suppose that there exists $j \in M$ and $k \in N$ where $p_{(ikj)} > 0$ and $y_{i(ikj)} > \min\limits_{\substack{k' \in N:\\ a_{ij} \neq a_{k'j}}} y_{i (ik'j)}$. Define the bundle $y'_i \sim R(\Gamma)$ to be identical to $y_i$, except that $y_{i(ikj)} = \min\limits_{\substack{k' \in N:\\ a_{ij} \neq a_{k'j}}} y_{i (ik'j)}$. Define the bundle $y''_i \sim R(\Gamma)$ by 
\[
y''_{i\ell} = y'_{i\ell} + \frac{\Big(y_{i(ikj)} - \min\limits_{\substack{k' \in N:\\ a_{ij} \neq a_{k'j}}} y_{i (ik'j)}\Big)p_{(ikj)}}{\sum_{\ell' \in R(M)} p_{\ell '}}
\]
Then $u_i^R(y''_i) > u_i^R(y'_i) = u_i^R(y_i)$, and $y''_i \cdot p = y_i \cdot p$. Thus $y_i$ cannot be in agent $i$'s demand set, which is a contradiction.
\end{proof}

Lemma~\ref{lem:reduction-demand-helper} is a straightforward combination of the previous two lemmas.
\begin{lemma}\label{lem:reduction-demand-helper}
Given prices $p\sim R(\Gamma)$, suppose a bundle $y_i  \sim R(\Gamma)$ is in $D_i^R(p)$. Then $\Rfrom(y_i) \cdot \Rfrom(p)_i = y_i \cdot p$.
\end{lemma}

\begin{proof}
By Lemma~\ref{lem:reduction-demand-pairwise}, we have (1) for any $j \in M$ and $k,k' \in N\backslash \{i\}$ where $y_{i(kk'j)} \neq 0$, we have $p_{(kk'j)} = 0$, and (2) for any $j \in M$ and $k\in N$ where $y_{i(ikj)} \neq \min\limits_{\substack{k \in N:\\ a_{ij} \neq a_{kj}}} y_{i (ikj)}$, we have $p_{(ikj)} = 0$. Therefore by Lemma~\ref{lem:reduction-cost}, we have $\Rfrom(y_i) \cdot \Rfrom(p)_i = y_i \cdot p$.
\end{proof}

Lemma~\ref{lem:reduction-demand} states that $y_i$ is in agent $i$'s demand set in $R(\Gamma)$ if and only if $\Rfrom(y_i)$ is in agent $i$'s demand set in $\Gamma$. This will not only play an important role in the proof of Theorem~\ref{thm:reduction-eq}, but also later on in t\^{a}tonnement.

The majority of the proof of Lemma~\ref{lem:reduction-demand} is devoted to proving that
\[
\max\limits_{\substack{y_i' \sim \Gamma:\\ y_i' \cdot \Rfrom(p)_i \leq B_i}} u_i(y_i')
= \max\limits_{\substack{y_i'\sim R(\Gamma): \\ y'_i \cdot p \leq B_i}} u_i^R(y_i')
\]

The intuitive argument for the above equality is that the utilities and costs of bundles are the same in both $\Gamma$ and $R(\Gamma)$. Slightly more formally, for any bundle $y_i \sim \Gamma$, $R(y_i) \sim R(\Gamma)$ has the same utility (by definition) and the same cost (by Lemma~\ref{lem:reduction-forward-cost}). For any bundle $y_i \sim R(\Gamma)$, $u_i^R(y_i) = u_i(\Rfrom(y_i))$ is also true by definition, but $y_i$ and $\Rfrom(y_i)$ do not necessarily have the same cost. That is where Lemma~\ref{lem:reduction-demand-helper} will be important.

\begin{lemma}\label{lem:reduction-demand}
Given prices $p \sim R(\Gamma)$ and a bundle $y_i \sim R(\Gamma)$, $y_i \in D_i^R(p)$ if and only if $\Rfrom(y_i) \in D_i(\Rfrom(p))$.
\end{lemma}

\begin{proof}
Lemma~\ref{lem:reduction-forward-cost} states that $y'_i \cdot \Rfrom(p)_i = R(y'_i) \cdot p$ for any bundle $y'_i \sim \Gamma$. This implies the following set equivalence:
\[
\{y'_i\ |\ y'_i\sim \Gamma\ \text{and}\ y'_i \cdot \Rfrom(p)_i \leq B_i\}
= \{y'_i\ |\ y'_i \sim \Gamma\ \text{and}\ R(y'_i) \cdot p \leq B_i\}
\]
Next, recall that for any bundle $y'_i \sim \Gamma$, $\Rfrom(R(y'_i)) = y'_i$, so
\[
\{y'_i\ |\ y'_i\sim \Gamma\ \text{and}\ y'_i \cdot \Rfrom(p)_i \leq B_i\}
= \{\Rfrom(R(y'_i))\ |\ y'_i \sim \Gamma\ \text{and}\ R(y'_i) \cdot p \leq B_i\}
\]
For every bundle $y'_i \sim \Gamma$, $R(y_i')$ is a bundle in $R(\Gamma)$. Therefore we can replace $R(y'_i)$ with $y'_i$ and get
\[
\{\Rfrom(R(y'_i))\ |\ y'_i \sim \Gamma\ \text{and}\ R(y'_i) \cdot p \leq B_i\}
\subseteq \{\Rfrom(y'_i)\ |\ y'_i \sim R(\Gamma)\ \text{and}\ y'_i \cdot p \leq B_i\}
\]
Note that the relationship is now subset instead of equality. This is because there may be some $y'_i \sim R(\Gamma)$ that does not equal $R(y''_i)$ for any $y''_i \sim \Gamma$. Combining this subset relationship with the previous equality gives us
\[
\{y'_i\ |\ y'_i\sim \Gamma\ \text{and}\ y'_i \cdot \Rfrom(p)_i \leq B_i\}
\subseteq \{\Rfrom(y'_i)\ |\ y'_i \sim R(\Gamma)\ \text{and}\ y'_i \cdot p \leq B_i\}
\]
Since $u_i(\Rfrom(y'_i)) = u_i^R(y'_i)$ by definition, we have
\[
\{u_i(y'_i)\ |\ y'_i\sim \Gamma\ \text{and}\ y'_i \cdot \Rfrom(p)_i \leq B_i\}
\subseteq \{u_i^R(y'_i)\ |\ y'_i \sim R(\Gamma)\ \text{and}\ y'_i \cdot p \leq B_i\}
\]
Taking the max gives us
\[
\max\Big(\{u_i(y'_i)\ |\ y'_i\sim \Gamma\ \text{and}\ y'_i \cdot \Rfrom(p)_i \leq B_i\}\Big)
\leq
\max\Big(\{u_i^R(y'_i)\ |\ y'_i \sim R(\Gamma)\ \text{and}\ y'_i \cdot p \leq B_i\}\Big)
\]
which we can rewrite as
\[
\max\limits_{\substack{y_i' \sim \Gamma:\\ y_i' \cdot \Rfrom(p)_i \leq B_i}} u_i(y_i')
\leq \max\limits_{\substack{y_i'\sim R(\Gamma): \\ y'_i \cdot p \leq B_i}} u_i^R(y_i')
\]
Consider an arbitrary $y_i'' \in D_i^R(p)$: then
\[
u_i^R(y_i'') = \max\limits_{\substack{y_i'\sim R(\Gamma): \\ y'_i \cdot p \leq B_i}} u_i^R(y_i')
\]
and $y_i'' \cdot p \leq B_i$. Then by Lemma~\ref{lem:reduction-demand-helper}, $\Rfrom(y_i'') \cdot \Rfrom(p)_i = y_i'' \cdot p \leq B_i$. Since $\Rfrom(y_i'') \sim \Gamma$ and $\Rfrom(y_i'') \cdot \Rfrom(p)_i \leq B_i$, we have
\[
\max\limits_{\substack{y_i' \sim \Gamma:\\ y_i' \cdot \Rfrom(p)_i \leq B_i}} u_i(y_i')
\geq u_i(\Rfrom(y_i''))
\]
By definition, $u_i^R(y_i'') = u_i(\Rfrom(y_i''))$, so
\[
\max\limits_{\substack{y_i'\sim R(\Gamma): \\ y'_i \cdot p \leq B_i}} u_i^R(y_i')
= u_i^R(y_i'') = u_i(\Rfrom(y_i'')) \leq 
\max\limits_{\substack{y_i' \sim \Gamma:\\ y_i' \cdot \Rfrom(p)_i \leq B_i}} u_i(y_i')
\leq \max\limits_{\substack{y_i'\sim R(\Gamma): \\ y'_i \cdot p \leq B_i}} u_i^R(y_i')
\]
Therefore,
\[
\max\limits_{\substack{y_i' \sim \Gamma:\\ y_i' \cdot \Rfrom(p)_i \leq B_i}} u_i(y_i')
= \max\limits_{\substack{y_i'\sim R(\Gamma): \\ y'_i \cdot p \leq B_i}} u_i^R(y_i')
\]

Finally, suppose $y_i \in D_i^R(p)$: then $y_i \cdot p \leq B_i$, and by Lemma~\ref{lem:reduction-demand-helper} we have $\Rfrom(y_i) \cdot \Rfrom(p)_i = y_i \cdot p \leq B_i$. Also,
\[
u_i(\Rfrom(y_i)) = u_i^R(y_i) = 
\max\limits_{\substack{y_i'\sim R(\Gamma): \\ y'_i \cdot p \leq B_i}} u_i^R(y_i')
= \max\limits_{\substack{y_i' \sim \Gamma:\\ y_i' \cdot \Rfrom(p)_i \leq B_i}} u_i(y_i')
\]
so $\Rfrom(y_i) \in D_i(\Rfrom(p))$. Suppose $\Rfrom(y_i) \in D_i(\Rfrom(p))$: then $\Rfrom(y_i) \cdot \Rfrom(p)_i \leq B_i$. Since $y_i = R(\Rfrom(y_i))$, we have $y_i \cdot p = \Rfrom(y_i) \cdot \Rfrom(p)_i \leq B_i$. Also,
\[
u_i^R(y_i) 
= u_i(\Rfrom(y_i))
= \max\limits_{\substack{y_i' \sim \Gamma:\\ y_i' \cdot \Rfrom(p)_i \leq B_i}} u_i(y_i')
= \max\limits_{\substack{y_i'\sim R(\Gamma): \\ y'_i \cdot p \leq B_i}} u_i^R(y_i')
\]
Therefore $y_i \in D_i^R(p)$.

\end{proof}

Recall that for any bundle $y_i \sim \Gamma$, $\Rfrom(R(y_i)) = y_i$. Thus by Lemma~\ref{lem:reduction-demand}, $R(y_i) \in D_i^R(p)$ if and only if $y_i = \Rfrom(R(y_i)) \in D_i(p)$. This is expressed by Corollary~\ref{cor:reduction-demand}, which will be useful in Section~\ref{sec:tat} when considering \tat processes in the PDM.

\begin{lemcorollary}[of Lemma~\ref{lem:reduction-demand}]\label{cor:reduction-demand}
Given prices $p\sim R(\Gamma)$ and a bundle $y_i \sim \Gamma$, $y_i \in D_i(\Rfrom(p))$ if and only if $R(y_i) \in D_i^R(p)$.
\end{lemcorollary}

\thmReductionEq*

\begin{proof}
$(\implies)$ Suppose $(\y, p)$ is a ME of the Fisher market $R(\Gamma)$: then $y_i \in D_i^R(p)$ for all $i \in N$, and $\sum\limits_{i \in N} y_{i\ell} = 1$ or $p_\ell = 0$ for all $\ell \in R(M)$. By Lemma~\ref{lem:reduction-demand}, we have $\Rfrom(y_i) \in D_i(\Rfrom(p))$.

We define $\x = (x^1...x^m) \in [0,1]^{m\times 2}$ as follows:
\begin{align}
x^{j,0} =&\ \max_{i\in N: a_{ij} = 0} \Rfrom(y_i)_j\label{eqn:xj0}\\\nonumber
x^{j,1} =&\ 1 - x^{j,0}
\end{align}
for all $j \in M$. We claim that for all $i \in N$ and $j \in M$, $\Rfrom(y_i)_j \leq x^{j,a_{ij}}$, and that $\Rfrom(y_i)_j = x^{j, a_{ij}}$ if $\Rfrom(p)_{ij} > 0$. 

We first show that $\Rfrom(y_i)_j \leq x^{j,a_{ij}}$ for all $i,j$. When $a_{ij} = 0$, this is true by definition, so assume $a_{ij} = 1$. Since $(\y, p)$ is a ME of $R(\Gamma)$, for any $\ell \in R(M)$, we have $\sum_{k' \in N} y_{k'\ell} \leq 1$. Thus for any $k \in N$ where $a_{ij} \neq a_{kj}$ (i.e. $a_{kj} = 0$), we have $\sum_{k' \in N} y_{k'(ikj)} \leq 1$.

Also, recall that by definition, $\Rfrom(y_i)_j = \min_{k\in N: a_{ij} \neq a_{kj}} y_{i(ikj)}$. Thus $\Rfrom(y_i)_j \leq y_{i(ikj)}$ for all $k$. Similarly, $\Rfrom(y_k)_j \leq y_{k(ikj)}$. Therefore
\begin{align}
\Rfrom(y_i)_j&\ + \Rfrom(y_k)_j \leq y_{i(ikj)} + y_{k(ikj)} \leq \sum\limits_{k' \in N} y_{k'(ikj)} \leq 1 \hspace{.2 in} \forall k \in N:\ a_{kj} = 0 \label{eqn:lessinitialline}\\
\Rfrom(y_i)_j&\ + \max\limits_{k \in N: a_{kj} = 0} \Rfrom(y_k)_j \leq 1\\
\Rfrom(y_i)_j&\ \leq 1 - \max\limits_{k \in N: a_{kj} = 0} \Rfrom(y_k)_j = x^{j,1} \label{eqn:lesslastline}
\end{align}
Thus $\Rfrom(y_i)_j \leq x^{j,a_{ij}}$ for all $i \in N$. It remains to show that $\Rfrom(y_i)_j = x^{j, a_{ij}}$ whenever $\Rfrom(p)_{ij} > 0$.

Suppose for sake of contradiction there exists $i \in N$ and $j \in M$ where $\Rfrom(y_i)_j < x^{j,a_{ij}}$ and $\Rfrom(p)_{ij} > 0$. Since $\Rfrom(y_i)_j = \min\limits_{\substack{k \in N: a_{ij} \neq a_{kj}}} y_{i(ikj)}$, there must exist $k \in N$ where $a_{ij} \neq a_{kj}$ and $y_{i(ikj)} < x^{j,a_{ij}}$. Since $\Rfrom(p)_{ij} = \sum\limits_{k' \in N: a_{ij} \neq a_{k'j}} p_{(ik'j)}$, there must exist $k'$ where $p_{(ik'j)} > 0$.

If $y_{i(ik'j)} > y_{i(ikj)}$, we have $p_{(ik'j)} = 0$ by Lemma~\ref{lem:reduction-demand-pairwise}. Thus assume $y_{i(ik'j)} \leq y_{i(ikj)} < x^{j,a_{ij}}$. We showed above that $\Rfrom(y_{k'})_j \leq x^{j,a_{k'j}} = 1 - x^{j,a_{ij}}$: thus $y_{k'(ik'j)} \leq 1-x^{j, a_{ij}}$. Therefore $y_{i(ik'j)} + y_{k'(ik'j}) < x^{j, a_{ij}} + 1 - x^{j, a_{ij}}  = 1$. If there exists $i' \not\in {i, k'}$ where $y_{i'(ik'j)} > 0$, then $p_{(ik'j)} = 0$, which is a contradiction. Therefore $\sum_{\ell \in R(M)} y_{i\ell} = y_{i(ik'j)} + y_{k'(ik'j}) < 1$.

But then by the definition of a ME, we have $p_{(ik'j)} = 0$, which is again a contradiction. Therefore $\Rfrom(y_i)_j = x^{j, a_{ij}}$ whenever $\Rfrom(p)_{ij} > 0$. This shows that $(\Rfrom(y), \Rfrom(p))$ is a PME of $\Gamma$.

$(\impliedby)$ Suppose $(\Rfrom(\y), \Rfrom(p))$ is a PME of the PDM $\Gamma$. Then $\Rfrom(y_i) \in D_i(\Rfrom(p))$ for all $i \in N$, so $y_i \in D_i^R(p)$ by Lemma~\ref{lem:reduction-demand}. Also, there exists $\x = (x^1...x^m) \in [0,1]^{m\times 2}$ where for all $i\in N$ and $j \in M$, 
\begin{enumerate}
\item $x^{j,0} + x^{j,1} = 1$
\item $\Rfrom(y_i)_j \leq x^{j, a_{ij}}$
\item $\Rfrom(y_i)_j = x^{j, a_{ij}}$ whenever $\Rfrom(p)_{ij} = 0$.
\end{enumerate}

It remains to show that for all $\ell \in R(M)$, either $\sum_{i\in N} y_{i\ell} = 1$ or $p_\ell = 0$. Suppose for sake of contradiction that there exists $\ell = (i,k,j) \in R(M)$ where $\sum_{k' \in N} y_{k'(ikj)} < 1$ and $p_{(ikj)} > 0$. If there exists $k' \not \in \{i,k\}$ where $y_{k'(ikj)} > 0$, then $p_{(ikj)} = 0$ by Lemma~\ref{lem:reduction-demand-pairwise}. Thus 
\[
\sum_{k' \in N} y_{k'(ikj)} = y_{i(ikj)} + y_{k(ikj)} < 1
\]
Furthermore, if either $y_{i(ikj)} \neq \Rfrom(y_i)_j$ or $y_{k(ikj)} \neq \Rfrom(y_k)_j$, we have $p_{(ikj)} = 0$ again by Lemma~\ref{lem:reduction-demand-pairwise}. Thus $y_{i(ikj)} = \Rfrom(y_i)_j$ and $y_{k(ikj)} = \Rfrom(y_k)_j$, so
\[
\Rfrom(y_i)_j +\Rfrom(y_k)_j < 1
\]

Recall that $\Rfrom(y_i)_j \leq x^{j, a_{ij}}$ and $\Rfrom(y_k)_j \leq x^{j, a_{kj}}$, and that $x^{j, a_{ij}} + x^{j, a_{kj}} = 1$ since $a_{ij} \neq a_{kj}$. Thus in order for $\Rfrom(y_i)_j +\Rfrom(y_k)_j < 1$ to be true, either $\Rfrom(y_i)_j < x^{j, a_{ij}}$ or $\Rfrom(y_k)_j < x^{j, a_{kj}}$. By symmetry, suppose $\Rfrom(y_i)_j < x^{j, a_{ij}}$ without loss of generality. Then because $(\Rfrom(\y), \Rfrom(p))$ is a PME, we have $\Rfrom(p)_{ij} = 0$.

By definition, $\Rfrom(p)_{ij} = \sum\limits_{\substack{k' \in N: a_{ij} \neq a_{k'j}}} p_{(ik'j)}$. Since $p_{(ik'j)} \geq 0$ for all $i,k',j$, we have $p_{(ik'j)}$ for all $k' \in N$ where $a_{ij} \neq a_{k'j}$. But then $p_{(ikj)} = 0$, which is a contradiction. Thus for all $\ell \in R(M)$, either $\sum_{i\in N} y_{i\ell} = 1$ or $p_\ell = 0$. Therefore $(\y, p)$ is a ME of $R(\Gamma)$.
\end{proof}

\thmReductionMax*

\begin{proof}
We first claim that $R(\x)$ is a valid allocation in $R(\Gamma)$, meaning that $\sum_{i \in N} R(x_i)_\ell \leq 1$ for all $\ell \in R(M)$. By definition of $R(x_i)$, $R(x_i)_{(kk'j)} = 0$ whenever $i \not\in \{k,k'\}$, and $R(x_i)_{(kk'j)} = x_{ij}$ whenever $i \in \{k,k'\}$. Therefore, for all $\ell \in R(M)$,
\[
\sum_{i \in N} R(x_i)_\ell = \sum_{i,k,k' \in N} R(x_i)_{(kk'j)}
= R(x_i)_{(ikj)} + R(x_k)_{(ikj)}  = x_{ij} + x_{kj}
\]
By definition of $R(M)$, the fact that good $(i,k,j)$ exists implies that $a_{ij} \neq a_{kj}$.  Since $\x$ is a valid outcome of $\Gamma$, for all $j\in M$ we must have $x_{ij} + x_{kj} \leq 1$ whenever $a_{ij} \neq a_{kj}$. Therefore $\sum_{i \in N} R(x_i)_\ell \leq 1$. Since this holds for all $\ell \in R(M)$, $R(\x)$ is a valid allocation in $R(\Gamma)$.

By definition of $\Rfrom$ and $u_i^R$, we have $u_i^R(x_i) = u_i(\Rfrom(x_i))$. Since $\Psi$ depends only on the agents' utilities, we have $\Psi(\mathbf{x'}) = \Psi(R(\mathbf{x'}))$ for any outcome $\mathbf{x'} \sim \Gamma$. Similarly, recall that $u_i(x_i') = u_i^R(R(x_i'))$ for any bundle $x_i' \sim \Gamma$, so $\Psi(\mathbf{x'}) = \Psi(\Rfrom(\mathbf{x'}))$ for any outcome $\mathbf{x'} \sim R(\Gamma)$. 

Thus for every possible outcome of $\Gamma$, there is an outcome of $R(\Gamma)$ which has the same value of $\Psi$, and for every possible outcome of $R(\Gamma)$, there is an outcome of $\Gamma$ which has the same value of $\Psi$. Therefore we have the numeral equality

\[
\max\limits_{\mathbf{x'} \sim \Gamma} \Psi(\mathbf{x'}) = \max\limits_{\mathbf{x'} \sim R(\Gamma)} \Psi(\mathbf{x'})
\]

Finally, because $\Psi(\x) = \Psi(R(\x))$, we have $\Psi(\x) \geq \alpha\cdot  \max\limits_{\mathbf{x'} \sim \Gamma} \Psi(\mathbf{x'})$ if and only if $\Psi(R(\x)) \geq \alpha \cdot \max\limits_{\mathbf{x'} \sim R(\Gamma)} \Psi(\mathbf{x'})$. Therefore $\x$ is an $\alpha$-approximation of $\Psi$ in $\Gamma$ if and only if $R(\x)$ is an $\alpha$-approximation of $\Psi$ in $R(\Gamma)$.
\end{proof}

\section{Other omitted proofs}\label{sec:omitted-proofs}

\subsection{Omitted proofs from Section~\ref{sec:per-issue-lower}}\label{sec:omitted-proofs-lower}

\linPerIssueExist*

\begin{proof}
Let $\Gamma$ be the PDM $(N,M, B)$ with linear utilities $u_i$ given by weights $w_{ij}$, and $\tilde{\Gamma}$ be the Fisher market $(N,M, B)$ with linear utilities $\tilde{u_i}$ given by the same weights.

Let $(\y, p)$ be an IME of $\Gamma$: then $y_i \in D_i(p, y_{-i})$ for all $i$. Let $x_i$ be agent $i$'s public bundle in $\y$, let $\mathbf{y'} = (y_{-i}, y'_i)$ for an arbitrary private bundle $y'_i$, and let $x'_i$ be agent $i$'s public bundle for private bundles $\mathbf{y'}$. Then we have

\begin{align*}
u_i(y_{-i}, y_i) =&\ \max\limits_{y_i':\ y_i' \cdot p \leq B_i} u_i(y_{-i}, y_i')\\
\sum\limits_{j \in M} w_{ij} x_{ij} =&\ \max\limits_{y_i':\ y_i' \cdot p \leq B_i} \sum\limits_{j \in M} w_{ij} x'_{ij}\\
\sum\limits_{j \in M} w_{ij} \sum\limits_{\substack{k \in N:\\ a_{kj} = a_{ij}}} y_{kj} =&\
	 \max\limits_{y_i':\ y_i' \cdot p \leq B_i} \Bigg(\sum\limits_{j \in M} w_{ij}y'_{ij} +\sum\limits_{j \in M}\ w_{ij} \sum\limits_{\substack{k \in N\backslash\{i\}:\\ a_{kj} = a_{ij}}} y_{kj}\Bigg)\\
\sum\limits_{j \in M} w_{ij}y_{ij} + \sum\limits_{j \in M}\ w_{ij} \sum\limits_{\substack{k \in N\backslash\{i\}:\\ a_{kj} = a_{ij}}} y_{kj} =&\
	 \max\limits_{y_i':\ y_i' \cdot p \leq B_i} \Bigg(\sum\limits_{j \in M} w_{ij}y'_{ij} +\sum\limits_{j \in M}\ w_{ij} \sum\limits_{\substack{k \in N\backslash\{i\}:\\ a_{kj} = a_{ij}}} y_{kj}\Bigg)\\
\sum\limits_{j \in M} w_{ij}y_{ij} + \sum\limits_{j \in M}\ w_{ij} \sum\limits_{\substack{k \in N\backslash\{i\}:\\ a_{kj} = a_{ij}}} y_{kj} =&\
	 \max\limits_{y_i':\ y_i' \cdot p \leq B_i} \Bigg(\sum\limits_{j \in M} w_{ij}y'_{ij}\Bigg) +\sum\limits_{j \in M}\ w_{ij}  \sum\limits_{\substack{k \in N\backslash\{i\}:\\ a_{kj} = a_{ij}}} y_{kj}\\
\sum\limits_{j \in M} w_{ij}y_{ij} =&\
	 \max\limits_{y_i':\ y_i' \cdot p \leq B_i} \sum\limits_{j \in M} w_{ij}y'_{ij}\\
\tilde{u_i}(y_i) =&\ \max\limits_{y_i':\ y_i' \cdot p \leq B_i} \tilde{u_i}(y_i')	 
\end{align*}

Also, the total price of $y_i$ is $y_i\cdot p$ in both $\Gamma$ and $\tilde{\Gamma}$. Let $\tilde{D_i}(p)$ be agent $i$'s demand set for prices $p$ in $\tilde{\Gamma}$: then by the above chain of equations, if $y_i \in D_i(p, y_{-i})$ for any $y_{-i}$, we have $y_i \in \tilde{D_i}(p)$. Furthermore, the exact same chain of equations in reverse order shows that if $y_i \in \tilde{D_i}(p)$, then $y_i \in D_i(p, y_{-i})$ for all $y_{-i}$.

Since $y_i \in \tilde{D_i}(p)$, the allocation $\y$ in $\tilde{\Gamma}$ gives each agent a bundle in her demand set given prices $p$. Also, because $(\y, p)$ is an IME of $\Gamma$, we have that $\sum_{i \in N} y_{ij} \leq 1$, and $\sum_{i \in N} y_{ij} = 1$ whenever $p_j > 0$. Therefore $(\y, p)$ is a ME of $\tilde{\Gamma}$.

Now let $(\y, p)$ be a ME of $\tilde{\Gamma}$. Since $y_i \in \tilde{D_i}(p)$ implies $y_i \in D_i(p, y_{-i})$, we have that $\y$ in $\Gamma$ gives each agent a bundle in her demand set. By the definition of ME, we have $\sum_{i \in N} y_{ij} \leq 1$, and $\sum_{i \in N} y_{ij} = 1$ whenever $p_j > 0$. Therefore $(\y, p)$ is an IME of $\Gamma$.
\end{proof}


\linPerIssueLower*

\begin{proof}
Let $\Phi'(n,1+\epsilon)$ be the Fisher market with the same agents, goods, and weights as $\Phi(n,1+\epsilon)$, also with linear utilities. Let $(\y, p)$ be a ME of $\Phi'(n,1+\epsilon)$. Then $y_i \cdot p = B_i = 1$ for all $i$, and so $\sum_{j \in M} p_j = \B = n$.

We next observe for a Fisher market with linear utilities, any (private) bundle in an agent's demand set maximizes her ``bang-per-buck" ratio: $w_{ij}/p_j$. To see this, consider agent $i$ moving $\delta$ of her budget to a good that does not maximize her bang-per-buck: this would decrease her utility, and so that bundle cannot be in her demand set.

Suppose for sake of contradiction that there exists $\ell$ where $p_\ell \neq 1$. Since $\sum_{j \in M} p_j = n$, there must exist $\ell$ where $p_\ell < 1$. Let $\ell = \min_{j \in M} p_j$. Since $w_{\ell\ell} > w_{\ell j}$ for all $j\neq \ell$, only issue $\ell$ maximizes agent $\ell$'s bang-per-buck. Thus there is a single bundle $y_i$ in her demand set, and it consists of her spending her entire budget on issue $\ell$. But since $p_\ell < 1 = B_\ell$, agent $\ell$ purchases more of good $\ell$ than exists in the supply, and so the market cannot clear. Thus any ME of $\Phi'(n,1+\epsilon)$ must have $p_j = 1$ for all $j$.

Now assume that $p_j = 1$ for all $j$. Since for each agent $i$, $w_{ii} = w > 1 = w_{ij}$ for all $j\neq i$, the only bundle that maximizes agent $i$'s bang-per-buck consists of her spending her entire budget on issue $i$. Thus the unique ME is $(\y, p)$ where $y_{ii} = 1$ for all $i$, and $y_{ij} = 0$ whenever $i\neq j$. Furthermore, by Theorem~\ref{thm:lin-per-issue-exist}, $(\y, p)$ is the unique IME of $\Phi(n,1+\epsilon)$. 

The Nash welfare of $\y$ in $\Phi(n,1+\epsilon)$ is
\[
NW(\y) =
\Big(\prod\limits_{i \in N} \sum\limits_{j\in M} w_{ij}y_{ij}\Big)^{1/n}
= \Big(\prod\limits_{i \in N} 1+\epsilon \Big)^{1/n} = 1+\epsilon
\]
Now consider the outcome $\z$ where $z^{j,0} = 0$ and $z^{j,1} = 1$ for all $j\in M$. Let $x_{ij}$ be agent $i$'s public bundle, as usual. Then
\[
NW(\z) 
=\Big(\prod\limits_{i \in N} \sum\limits_{j\in M} w_{ij}x_{ij}\Big)^{1/n}
= \Big(\prod\limits_{i \in N} \sum\limits_{j\in M\backslash\{i\}} 1 \cdot x_{ij}\Big)^{1/n}
= \Big(\prod\limits_{i \in N} (n-1) \Big)^{1/n} = n-1
\]
and therefore
\[
\frac{\max\limits_{\mathbf{z'}} NW(\mathbf{z'})}{NW(\y)} \geq \frac{n-1}{1+\epsilon} 
\]
\end{proof}


We now present the proofs of Theorems~\ref{thm:cd-per-issue-lower} and \ref{thm:ces-per-issue-lower}, which state that issue-pricing equilibria can be inefficient for Cobb-Douglas and CES utilities, respectively. 

We first prove a lemma motivated by the following concept. In Section~\ref{sec:lin-per-issue-lower}, we described how for linear utilities, the bundles in an agent's demand set maximize her bang-per-buck, in both the public and private settings. This is not true in general for other utilities, since goods are not independent. However, the same concept still applies: agent $i$ will not spend any money on issue $j$ if there is another issue $\ell$ where she has a higher marginal utility per dollar spent on issue $\ell$. This concept will be made formal by examining $\mfrac{\partial(u_i(x_i))}{\partial x_{ij}}$, which is the partial derivative of agent $i$'s utility with respect to $x_{ij}$. Although these derivatives may be complicated in general, they are well-behaved for Cobb-Douglas and CES utilities with $\rho \in (-\infty, 0) \cup (0, 1)$.

We will use this concept to show that for Cobb-Douglas utilities, $x_{ij} p_j = \min_{\ell \in M} x_{i\ell} p_\ell$ for any issue $j$ that agent $i$ is spending any money on. For CES utilities with $\rho \in (-\infty, 0) \cup (0, 1)$, we will show that $x_{ij}^{1-\rho} p_j = \min_{\ell \in M} x_{i\ell}^{1-\rho} p_\ell$ (note that $1- \rho > 0$ since $\rho \in (-\infty, 0) \cup (0, 1)$). Using these two properties, the following lemma will imply that $x_{ij} = 1/2$ for all $j$, which allows us to compute the Nash welfare.

\begin{restatable}{lemma}{otherUtilitiesEq}
\label{lem:other-utilities-eq}
Let $(\y ,p)$ be an IME of $\Phi(n,1)$ and let $x_i$ be agent $i$'s public bundle as induced by $\y$. Suppose that there exists $c > 0$ such that for every issue $j$ that agent $i$ spends any money on, $x_{ij}^c p_j = \min\limits_{\ell \in M} x_{i\ell}^c p_\ell$. Then $x_{ij} = 1/2$ for all $i$ and $j$.
\end{restatable}

\begin{proof}
The majority of the proof will be dedicated to showing that for every agent $i$, there must exist an issue $j$ where $x_{ij}^cp_j \leq 1/2^c$. Suppose for sake of contradiction that there exists an agent $i$ where $x_{ij}^c p_j > 1/2^c$ for every issue $j$.

We first show that there must exist an agent $k$ and issue $j$ where $x_{kj}^c p_j < 1/2$. Because $(\y, p)$ is an IME, all agents exhaust their budgets, so $\sum_{j \in M} p_j = \sum_{k \in N} B_k = n$. Because $|M| = n$ here, there must exist $j \in M$ where $p_j \leq 1$. Since $x_{ij}^c p_j > 1/2^c$, we have $x_{ij} > 1/2$. Let $k$ be any agent where $a_{kj} \neq a_{ij}$: then $x_{kj} < 1/2$, and so $x_{kj}^c p_j < 1/2^c$.

We know by definition of $\Phi(n,1)$, agents $i$ and $k$ agree on all issues other than $i$ and $k$: $a_{i \ell} = a_{k\ell}$ whenever $\ell \not \in \{i,k\}$. Thus for all issues $\ell \not \in \{i,k\}$, $x_{k\ell}^c p_\ell > 1/2^c > x_{kj}^c p_j$. By assumption, agent $k$ only spends money on issues $\ell$ which minimize $x_{k \ell}^c p_\ell$. Thus agent $k$ does not spend money on any issues besides $i$ and $k$ (note that either $j = i$ or $j = k$). 

Therefore amount of money agent $k$ spends in total is $\sum_{\ell \in M} y_{k\ell}p_\ell = y_{kk}p_k + y_{ki}p_i$. Since agent $k$ exhausts her budget, we have $y_{kk}p_k + y_{ki}p_i = B_k = 1$. Thus there must exist $\ell \in \{k,i\}$ where $y_{k\ell}p_\ell \geq 1/2$. Therefore $x_{k\ell}p_\ell \geq 1/2$.

Since $a_{kk} \neq a_{ik}$ and $a_{ki} \neq a_{ii}$, we have $a_{k\ell} \neq a_{i\ell}$. Because $(\y, p)$ is an IME, we have $x_{i\ell} = 1 - x_{k\ell}$. Also, since agent $k$ spends money on issue $\ell$, we have $x_{k\ell}^c p_\ell \leq x_{kj}^c p_j < 1/2^c$ by assumption. Therefore
\begin{align*}
x_{k\ell}^c p_\ell <&\ 1/2^c < x_{i\ell}^c p_\ell\\
x_{k\ell}^cp_\ell <&\ (1-x_{k\ell})^c p_\ell \\
x_{k\ell}^c <&\ (1-x_{k\ell})^c\\
x_{k\ell} <&\ (1-x_{k\ell})\\
x_{k\ell} <&\ 1/2
\end{align*}
Since agent $k$ exhausts her budget, we have $y_{kk}p_k + y_{ki}p_i = B_k = 1$, and so $x_{kk}p_k + x_{ki}p_i = 1$. Thus there must exist $\ell \in \{i, k\}$ where $y_{k\ell}p_\ell \geq 1/2$. Because $x_{k\ell}^c p_\ell < 1/2^c$, we have $x_{k\ell}^{1-c}/2^c > x_{k\ell}p_\ell \geq 1/2$. Therefore
\begin{align*}
\cfrac{x_{k\ell}^{1-c}}{2^c} >& \cfrac{1}{2}\\
\cfrac{(1/2)^{1-c}}{2^c} >&\ \cfrac{1}{2}\\
1 >&\ 1
\end{align*}
which is a contradiction. Therefore for every agent $i$, there exists an issue $j$ where $x_{ij}^c p_j \leq 1/2^c$. 

By assumption, if $x_{ij}^c p_j > \min\limits_{\ell \in M} x_{i\ell}^c p_\ell$, then agent $i$ spends no money on issue $j$. Since there exists an issue $j$ where $x_{ij}^c p_j \leq 1/2^c$, we have that agent $i$ spends no money on any issue $j$ where $x_{ij}^c p_j > 1/2^c$. 

Suppose for sake of contradiction that an agent $i$ and issue $j$ exist where $x_{ij}^c p_j > 1/2^c$: then some agent $k$ where $a_{kj} = a_{ij}$ is spending money on issue $j$. But since $a_{kj} = a_{ij}$, we have $x_{kj}^c p_k > 1/2^c$, so agent $k$ cannot be spending any money on issue $j$. Therefore for every agent $i$ and every issue $j$, $x_{ij}^c p_j \leq 1/2^c$.

Suppose for sake of contradiction that there exists an issue $j$ where $p_j \neq 1$. Since $\sum_{\ell\in M} p_\ell = n$, there must exist an issue $\ell$ where $p_\ell > 1$. Let $k$ be any other agent other than $\ell$: then $a_{k\ell} \neq a_{\ell \ell}$. Since $x_{k\ell} + x_{\ell \ell} = 1$, we have $\max(x_{k\ell}, x_{\ell \ell}) \geq 1/2$. Therefore $(\max(x_{k\ell}, x_{\ell \ell}))^c p_\ell > 1/2^c$, which is a contradiction. Therefore $p_j = 1$ for all $j$.

Finally, suppose there exists an agent $i$ and issue $j$ where $x_{ij} \neq 1/2$, there must exist an agent $k$ where $x_{kj} > 1/2$. Then $x_{kj}^c p_j > 1/2^c$, which is again a contradiction. Therefore for every agent $i$ and issue $j$, $x_{ij} = 1/2$.
\end{proof}


We are now ready to prove Theorems~\ref{thm:cd-per-issue-lower} and \ref{thm:ces-per-issue-lower}.

\cdLower*

\begin{proof}
Let $x_i$ be agent $i$'s public bundle as induced by $\y$. Recall that a Cobb-Douglas utility is given by
\[
u_i(\y) = u_i(x_i) = \Big( \prod\limits_{j\in M} x_{ij}^{w_{ij}} \Big)^{1/\sum_{j \in M}w_{ij}}
\]
which for $\Phi(n,1)$, simplifies to
\[
u_i(x_i) = \Big( \prod\limits_{j\in M} x_{ij} \Big)^{1/n}
\]
Thus for all $j$, we have
\[
\cfrac{1}{p_j}\cfrac{\partial(u_i(\y))}{\partial x_{ij}} 
= \cfrac{x_{ij}^{\frac{1}{n}-1}}{p_j n} \Big( \prod\limits_{\ell \in M\backslash\{j\}} 
	x_{i\ell} \Big)^{1/n}
= \cfrac{1}{x_{ij} p_j n} \Big( \prod\limits_{\ell \in M} 
	x_{i\ell} \Big)^{1/n}
= \frac{1}{x_{ij} p_j n}u_i(x_i)
\]
We are going to invoke Lemma~\ref{lem:other-utilities-eq} with $c = 1$. Suppose that there exists an agent $i$ and issues $j,\ell$ such that agent $i$ is spending on issue $j$, but $x_{ij} p_j >  x_{i\ell} p_\ell$. Then
\[
\frac{1}{p_j}\cfrac{\partial(u_i(\y))}{\partial x_{ij}} < \frac{1}{p_\ell}\cfrac{\partial(u_i(\y))}{\partial x_{i\ell}}
\]
Thus there exists some $\delta > 0$ such that if agent $i$ spent $\delta$ less on issue $j$ and $\delta$ more on issue $\ell$, agent $i$'s utility would increase. But $x_i$ is in agent $i$'s demand set, so it cannot be possible for her to increase her utility while staying within her budget. This is a contradiction, so therefore $x_{ij} p_j = \min_{\ell \in M} x_{i\ell} p_\ell$ for all $i,j$.

Therefore by Lemma~\ref{lem:other-utilities-eq}, we have $x_{ij} = 1/2$ for all $i$ and $j$. So the Nash welfare of $\y$ is
\[
NW(\y) = \Big(\prod\limits_{i \in N} \Big( \prod\limits_{j\in M} x_{ij} \Big)^{1/n}\Big)^{1/n}
= \Big(\prod\limits_{i \in N} \Big( \prod\limits_{j\in M} 1/2 \Big)^{1/n}\Big)^{1/n}
= \Big(\prod\limits_{i \in N} 1/2 \Big)^{1/n}
= 1/2
\]
Consider the outcome $\mathbf{z'}$ where $x'_{ii}(\mathbf{z'}) = 1/n$ for all $i$, and $x'_{ij}(\mathbf{z'}) = \mfrac{n-1}{n}$ whenever $j\neq i$. Then
\[
NW(\mathbf{z'}) 
= \Big(\prod\limits_{i \in N} \Big( \prod\limits_{j\in M} x'_{ij} \Big)^{1/n}\Big)^{1/n}
= \Big(\prod\limits_{i \in N} \Big(\mfrac{1}{n}\Big(\mfrac{n-1}{n}\Big)^{n-1}\Big)^{1/n}\Big)^{1/n}
= \Big(\mfrac{1}{n}\Big(\mfrac{n-1}{n}\Big)^{n-1}\Big)^{1/n}
\]
\[
= \Big(\cfrac{1}{n-1}\Big(\cfrac{n-1}{n}\Big)^n\Big)^{1/n}
= \cfrac{1}{(n-1)^{1/n}}\cfrac{n-1}{n}
= \cfrac{1 - 1/n}{(n-1)^{1/n}}
\]
Therefore
\[
\frac{\max\limits_{\mathbf{z'}} NW(\mathbf{z'})}{NW(\y)} \geq \frac{2 - 2/n}{(n-1)^{1/n}}
\]

\end{proof}


\cesLower*

\begin{proof}
Let $x_i$ be agent $i$'s public bundle as induced by $\y$. Recall that a CES utility is given by
\[
u_i(\y) = u_i(x_i) = \Big(\sum\limits_{j \in M} w_{ij}^{\rho} x_{ij}^{\rho}\Big)^{1/\rho}
= \Big(\sum\limits_{j \in M} x_{ij}^{\rho}\Big)^{1/\rho}
\]
and so we have
\[
\cfrac{1}{p_j}\cfrac{\partial(u_i(\y))}{\partial x_{ij}} 
= \frac{1}{p_j}\cfrac{1}{\rho} \rho x_{ij}^{\rho - 1} \Big(\sum\limits_{j \in M} x_{ij}^{\rho}\Big)^{\frac{1}{\rho} - 1}
= \cfrac{1}{x_{ij}^{1 - \rho}p_j} u_i(x_i)^{\rho(\frac{1}{\rho} - 1)}
\]
This time, we are going to invoke Lemma~\ref{lem:other-utilities-eq} with $c = 1 - \rho$ (note that since $\rho \in (-\infty, 0) \cup (0, 1)$, we have $1 - \rho > 0$). Suppose that there exists an agent $i$ and issues $j,\ell$ such that agent $i$ is spending on issue $j$, but $x_{ij}^{1-\rho} p_j >  x_{i\ell}^{1-\rho} p_\ell$. Then
\[
\frac{1}{p_j}\cfrac{\partial(u_i(\y))}{\partial x_{ij}} < \frac{1}{p_\ell}\cfrac{\partial(u_i(\y))}{\partial x_{i\ell}}
\]
So again there exists some $\delta > 0$ such that if agent $i$ spent $\delta$ less on issue $j$ and $\delta$ more on issue $\ell$, agent $i$'s utility would increase. But $x_i$ is in agent $i$'s demand set, so this is a contradiction for the same reason as in the previous proof. Thus $x_{ij} p_j = \min_{\ell \in M} x_{i\ell} p_\ell$ for all $i,j$.

Therefore, by Lemma~\ref{lem:other-utilities-eq}, we have $x_{ij} = 1/2$ for all $i$ and $j$, so the Nash welfare of $\y$ is
\[
NW(\y) = \Big(\prod\limits_{i \in N} \Big( \sum\limits_{j\in M} x_{ij}^\rho \Big)^{1/\rho}\Big)^{1/n}
= \Big(\prod\limits_{i \in N} \Big( \sum\limits_{j\in M} (1/2)^\rho \Big)^{1/\rho}\Big)^{1/n}
= \Big(\prod\limits_{i \in N} \cfrac{n^{1/\rho}}{2} \Big)^{1/n}
= \cfrac{n^{1/\rho}}{2}
\]
Consider the outcome $\mathbf{z'}$ where $x'_{ii}(\mathbf{z'}) = 0$ for all $i$ and $x'_{ij}(\mathbf{z'}) = 1$ whenever $j\neq i$. Then
\[
NW(\mathbf{z'}) 
= \Big(\prod\limits_{i \in N} \Big( \sum\limits_{j\in M} x_{ij}^\rho \Big)^{1/\rho}\Big)^{1/n}
= \Big(\prod\limits_{i \in N} \Big( \sum\limits_{j\in M\backslash\{i\}} 1 \Big)^{1/\rho}\Big)^{1/n}
= \Big(\prod\limits_{i \in N} (n-1)^{1/\rho}\Big)^{1/n} = (n-1)^{1/\rho}
\]
Therefore
\[
\frac{\max\limits_{\mathbf{z'}} NW(\mathbf{z'})}{NW(\y)} \geq \frac{(n-1)^{1/\rho}}{n^{1/\rho}/2} = 2 \Big( \frac{n-1}{n} \Big)^{1/\rho} = 2 (1 - 1/n)^{1/\rho}
\]

\end{proof}


\subsection{Omitted proofs from Section~\ref{sec:tat}}\label{sec:omitted-proofs-tat}


\tatLift*

\begin{proof}
	By the reduction defined in Section~\ref{sec:reduction}, the hidden private market has $O(n^2m)$ goods (1 copy of each good for each pair of agents who disagree on the issue) and $n$ agents. 
	
	We next show the \tat $\T$ is being run correctly, i.e., the sequence of prices $(p^0, p^1, p^2...)$, alongside some demands $(y^t_i \in D_i^R(p^t))$ converges to a $\delta$-equilibrium. This is not trivial since $\T$ is run with a detour through the PDM. By Corollary~\ref{cor:reduction-demand}, given prices $p\sim R(\Gamma)$ and a bundle $y_i \sim \Gamma$, $y_i \in D_i(\Rfrom(p))$ $\iff$ $R(y_i) \in D_i^R(p)$. Thus at each step, $p^{t+1}$ is being calculated by $g_{\T}$ based on valid demands in $D_i^R(p^t)$, so the sequence of prices $(p^0, p^1, p^2...)$ converges to a $\delta$-equilibrium. Then, by supposition, there exist demands $y \sim R(\Gamma)$ at time $T$ such that $(\mathbf{y}, p^T) \sim R(\Gamma)$ form a $\delta$-equilibrium in the hidden Fisher market, for $T = O(\kappa(n^2m,n,\delta))$. 
	
	Recall Theorem~\ref{thm:reduction-eq}: $(\y, p)$ is a Fisher market equilibrium if and only if $(\Rfrom(\y), \Rfrom(p))$ is a PME. The rest of the proof involves showing that Theorem~\ref{thm:reduction-eq} holds for approximate equilibria as well, as defined. Recall that $\Rfrom(y_i)_j = \min\limits_{\substack{k \in N:\\ a_{ij} \neq a_{kj}}} y_{i (ikj)}$ $\forall j\in M$, and $\Rfrom(p)_{ij} = \sum\limits_{\substack{k \in N:\\ a_{ij} \neq a_{kj}}} p_{(ikj)}$ $\forall i\in N, j\in M$.
We claim $(\Rfrom(y), \Rfrom(p^t))$ forms a $3\delta$-PME:
	\begin{enumerate}
		\item  $y_i \in D_i^R(p^t) \implies \Rfrom(y_i) \in D_i(\Rfrom(p^t))$. (Lemma~\ref{lem:reduction-demand})
		\item We define $\z = (z^1...z^m) \in [0,1]^{m\times 2}$ as follows, analogously to Equation~\eqref{eqn:xj0} in the proof of Theorem~\ref{thm:reduction-eq}\footnote{In the proof of Theorem~\ref{thm:reduction-eq}, the definition of $z^{j,1}$ was simply $1 - z^{j,0}$. It is necessary to use $\max(1 - z^{j,0}, 0)$ here instead: because this is an approximate equilibrium, it is possible that $z^{j,0} > 1$, which would make $1 - z^{j,0}$ negative.}:
		\begin{align*}
		z^{j,0} =&\ \max_{i\in N: a_{ij} = 0} \Rfrom(y_i)_j\\
		z^{j,1} =&\ \max(1 - z^{j,0}, 0)
		\end{align*} Then, $\forall j\in M$,
		\begin{enumerate}
			\item $z^{j,0} + z^{j,1} \leq 1+\delta$ follows from the definition and from $y$ part of a $\delta$-equilibrium of a Fisher market.
			\item For all $i\in N$, $\Rfrom(y_i)_j \leq z^{j, a_{ij}}+ \delta$ follows from Equations~\eqref{eqn:lessinitialline}-\eqref{eqn:lesslastline}, with $1$ replaced with $1+\delta$ and the $=$ in line~\eqref{eqn:lesslastline} replaced with $\leq$. Finally,
			\begin{align*}
			\Rfrom(p)_{ij} > n\delta &\implies \exists \tilde{k} \text{ s.t. } p_{(ikj)} > \delta\\
			&\implies y_{i(i\tilde{k}j)} + y_{\tilde{k}(i\tilde{k}j)} > 1 - \delta & \text{(Condition~\ref{con:fisherdelp} of Fisher $\delta$-equilibrium)}
			\end{align*}
			By Lemma~\ref{lem:reduction-demand-pairwise}, $p_{(i\tilde{k}j)} > 0 \implies y_{i(i\tilde{k}j)} = \min\limits_{\substack{k \in N:\\ a_{ij} \neq a_{\tilde{k}j}}} y_{i (i\tilde{k}j)} = \Rfrom(y_i)_j$, and $y_{\tilde{k}(i\tilde{k}j)} = \Rfrom(y_{\tilde{k}})_j$. Then
			\begin{align*}
			\Rfrom(y_i)_j &> 1- \Rfrom(y_{\tilde{k}})_j - \delta \\
			&> 1-z^{j,a_{\tilde{k}j}} - 2\delta & \text{(First part of Condition (b)}\\
			&> z^{j,a_{ij}} - 3\delta & \text{definition of $z^{j,a_{ij}}$)}
			\end{align*}
			Thus, $\Rfrom(p)_{ij} > n\delta \implies \Rfrom(y_i)_j> z^{j,a_{ij}} - 3\delta$
		\end{enumerate}\label{condition:market-clears}
	\end{enumerate}

\end{proof}


\section{General t\^{a}tonnement with asymptotic convergence in Fisher markets}
\label{sec:tatprivatedetour}
The prior work discussed in Section~\ref{sec:tat} established deterministic t\^{a}tonnements with polynomial time convergence rates only for certain classes of utility functions, or for those that converge to weak equilibria. However, we would like a general PDM converges for a wider class of utilities (in particular, linear utilities). To achieve this, we sacrifice convergence in polynomial time (or any characterization of convergence rates), which has been the primary focus on prior work such as \cite{avigdor-elgrabli_convergence_2014,cheung_tatonnement_2013}\footnote{Those works take great care to show conditions analogous to strong convexity or Lipschitz continuity of the gradient in the cases of interest.}. In this section, we present a discrete stochastic gradient descent style t\^{a}tonnement for all Fisher markets that result from PDMs with EG utility functions through the reduction. This can then be lifted through Theorem~\ref{thm:tatonnementlifting} to yield a general PDM \tat.

This section broadly follows the gradient descent framework for Fisher market t\^{a}tonnements from~\cite{cheung_tatonnement_2013}, and the \tat can be seen as an asymptotic discretization of their continuous time \tat. The t\^{a}tonnement operates on the dual of the EG convex program, which is a function of the prices, and whose gradient is the excess demand (Lemma~\ref{lem:tatgradient}). We first establish that there exists a bounded, convex region $\Pi$ in which demands are bounded, with $p^*\in \Pi$\footnote{It is well known that the primal objective function is strictly concave, and so $p^*$, the optimal dual solution, is unique.}. We finish the proof with a standard SGD convergence technique, Lemma~\ref{lem:sgd}.

 Let $\phi(p)$ be the objective of the dual of EG convex program. We use $D^R$ to denote demands as the demands are in a Fisher market $R(\Gamma)$ that is constructed from PDM $\Gamma$. The following lemma from~\cite{cheung_tatonnement_2013} establishes that $\phi(p)$ is itself a convex function whose gradient is the excess demand.

\begin{lemma} [\cite{cheung_tatonnement_2013}]
	$\nabla \phi(p) = 1 - \sum_i D^R_i(p_i)$ \label{lem:tatgradient}
\end{lemma}
Note that $\phi(p)$ refers to the set of sub-gradients, and $D^R_i(p_i)$ to the set of demands. Even when demands at a given price are not unique (such as with linear utilities), each combination of demands yields a sub-gradient of the dual objective function. 

Before being able to apply a canonical gradient descent convergence theorem, we need to establish that there exists a bounded, convex set which contains the optimal price $p^*$ in its interior. We construct such a set next.

\begin{lemma} $\exists \Pi\subset \bbR_+^m$ bounded and convex s.t. $p^* \in \arg\max \phi(p)\subset \Pi$, and that $\forall p\in \Pi, \forall i$, $D^R_i(p) < \infty$.  \label{lem:pidemandsbounded}
	\end{lemma}
\begin{proof}
	We claim $p^* \in [0,2]^m$ in our setting. Let $y_{i} \in D^R_i(p^*_i)$, and thus $p^*_j y_{ij} \leq B_i = 1\, \forall j$. By Fisher market equilibrium conditions, $p^*_{j}>0 \implies \sum_i y_{ij} = 1$. In our setting, $y_{ij}>0$ for at most two distinct $i$. Thus, $\exists i$ s.t. $y_{ij} \geq \frac{1}{2} \implies p^*_j \leq 2$. 
	
	Let $p_{\min}$ be any value such that $ 0 < p_{\min} < \min_{\{j:p^*_j>0\}} p^*_j$. Then, let $\Pi = [p_{\min}^1, 2] \times \dots \times [p_{\min}^m, 2]$, where $p_{\min}^j = \begin{cases}
	0 & p^*_j=0\\
	p_{\min} & \text{else}
	\end{cases}$. $\Pi$ as defined has the desired properties.
	\end{proof} 

Throughout, we use $[\cdot]_\mathcal{X}$ denote the projection onto a set $\mathcal{X}$.
We will also use the following stochastic gradient descent convergence lemma.
\begin{lemma}[\cite{jiang_scheduling_2010}]
	Consider a convex function $f$ on a non-empty bounded closed convex set $\mathcal{X} \subset \bbR^m$, and use $[\cdot]_\mathcal{X}$ to designate the projection operator. Starting with some $x_0\in \mathcal{X}$, consider the SGD update rule $x_{t} = [x_{t-1} - \eta_t (\nabla f(x_t) + z_t + e_t)]_\mathcal{X}$, where $z_t$ is a zero-mean random variable and $e_t$ is a constant. Let $\textup{E}_t[\cdot]$ be the conditional expectation given $\mathcal{F}_t$, the $\sigma$-field generated by $x_0,x_1,\dots,x_t$. If
	\begin{align*}
	&f(\cdot)\text{ has a unique minimizer }x^* \in \mathcal{X}\\
	&\eta_t > 0, \sum _t \eta_t = \infty, \sum _t \eta_t^2 < \infty\\
	&\exists C_1 \in \bbR < \infty \text{ s.t. }\|\nabla f(x)\|_2 \leq C_1,\forall\, x \in \mathcal{X} \\
	&\exists C_2 \in \bbR < \infty \text{ s.t. } \textup{E}_t[\|z_t\|^2] \leq C_2, \forall\, t \\
	&\exists C_3 \in \bbR < \infty \text{ s.t. }\|e_t\|_2 \leq C_3,\forall\, t \\
	&\sum_{t} \eta_t \|e_t\| < \infty\text{ w.p. }1
	\end{align*}
	Then $x_t \to x^*$ w.p. 1 as $t \to \infty$.  
	\label{lem:sgd}
\end{lemma}

We can now construct a stochastic t\^{a}tonnement for a Fisher Market $R(\Gamma)$ that is constructed from PDM $\Gamma$ when agents have any EG utility function in the PDM, as $\mathcal{H}$-nested leontief utility functions remain EG utility functions.




\begin{lemma}
	Suppose $p^* \in \arg\min \phi(p)$. Then, $\exists y = (y_1 \dots y_n)$ s.t. $y_i \in D^R_i (p^*)$, $(y, p^*)$ is a ME. \label{lem:alldemandsinequil}
	\end{lemma}
\begin{proof}
Follows directly from Part 2 of Lemma 5 in~\cite{cheung_tatonnement_2013}, that $\arg\max_{x\geq0} L(x, p) \subseteq D^R(p)$. The ME $(x^*, p^*)$ is such that $x^* \in \arg\max_{x\geq0} L(x, p^*) \subseteq D^R(p^*)$. 
\end{proof}


\begin{theorem} Let agents $i$ in the PDM $\Gamma$ have utilities $u_i\in\mathcal{H}$ that are concave, continuous, non-decreasing, non-constant, and homogeneous of degree 1. Then there exists a stochastic gradient descent-style t\^{a}tonnement for which, as $t\to\infty$, $p^t \to p^*$, where $p^* \in \arg\min \phi(p)$, and $\exists x$ s.t. $x_i \in D_i(p^*)$ and $(x, \Rfrom(p^*))$ is a PME.	\label{thm:tatonnementgd}
\end{theorem}

	\begin{proof}
	Construct a Fisher market $R(\Gamma)$ through the reduction.
	
	Let $\T$ be the following descent in the constructed market. Start with prices $p^0\in \Pi$, where $\Pi$ as defined in Lemma \ref{lem:pidemandsbounded}. Update prices using the rule $p^{t+1} = [p^t - \eta_t\left(1 - \sum_i \tilde{D}^R_i(p^t_i)\right)]_{\Pi}$, for ${\eta_t} = \frac{1}{t}$, and $\tilde{D}^R_i(p) = y_{i,t} + b_{i,t} + z_{i,t}$, for some $y_{i,t} \in D^R_i(p^t)$. Assume $z_{i,t}$ a zero-mean random variable and $b_{i,t}$ a constant that follow the conditions of Lemma~\ref{lem:sgd}. 
	
	By Lemma~\ref{lem:nestedhomog}, the implied utility functions still yield Eisenberg-Gale markets and by Lemma~\ref{lem:tatgradient}, $\nabla \phi(p) = 1 - \sum_i D^R_i(p)$. By Lemma~\ref{lem:pidemandsbounded}, $\exists C<\infty~s.t.~\|\nabla \phi(p)\| < C, \forall p\in \Pi$. Convergence to prices $p^* \in \arg\min \phi(p)$ follows from Lemma~\ref{lem:sgd}. By Lemma~\ref{lem:alldemandsinequil}, $\exists \textbf{y}$, $y_i \in D^R_i (p^*)$ s.t. $(\textbf{y}, p^*)$ is a ME in the Fisher market. Thus, $\T$ converges asymptotically to a ME.

	 By Theorem~\ref{thm:tatonnementlifting}, $\T$ can be lifted to create a \tat $\Rfrom(\T)$ in the PDM that converges asymptotically to a PME.  
	

\end{proof}



Note that $\Pi$ is not known a priori. However, it can be approximated during the gradient descent without affecting convergence: for example, if at any point demand goes to infinity, backtrack and impose a minimum price. Then, if demand goes to $0$ with a positive price, lower this minimum price.



Theorem~\ref{thm:tatonnementgd} and Theorem~\ref{thm:tatonnementlifting} together create a t\^{a}tonnement with asymptotic convergence for Public Decision Markets for general concave, continuous, non-constant, non-decreasing, and homogeneous of degree 1 utility functions. 




\end{document}